\documentclass[12pt]{article}
\usepackage{amssymb}
\usepackage{amsmath}
\usepackage{amsfonts}
\usepackage{geometry,enumerate,bbm}
\usepackage{natbib}
\usepackage{epsfig}
\usepackage{float}
\usepackage[onehalfspacing]{setspace}
\usepackage{graphicx}
\usepackage{subcaption}
\usepackage{lscape}
\usepackage[flushleft]{threeparttable}
\usepackage{ulem}
\usepackage{scrextend}
\usepackage{xcolor}
\usepackage{braket}
\usepackage{tikz}

\usepackage{etoolbox}

\RequirePackage[colorlinks=true,citecolor=blue,urlcolor=blue,linkcolor=blue,breaklinks=true]{hyperref}

\newtheorem{theorem}{Theorem}
\newtheorem{corollary}{Corollary}

\newenvironment{proof}[1][Proof]{\noindent \textbf{#1.} }{\  \rule{0.5em}{0.5em}}
\newcommand\IgnoreThisText[1]{}

\begin{document}

\title{{\bf Unified Quantum Dynamics: \\ The Emergence of the Born Rule}}
\author{\setcounter{footnote}{0}
 Martin Weidner%
\thanks{
University of Oxford and Nuffield College.
Email: \texttt{martin.weidner@economics.ox.ac.uk} } }
\date{May 2025}

\maketitle
\thispagestyle{empty}
\setcounter{page}{0}

\begin{abstract}
\noindent
While the Born rule is traditionally introduced as a separate postulate of quantum mechanics, we show it emerges naturally from a modified Schrödinger equation that includes ``small-signal truncation''. This parallels the way quantum decoherence gives rise to the branches of the multiverse in the ``Many-Worlds Interpretation'', eliminating the need for a separate measurement postulate. Our approach thus offers a unified framework in which both the emergence of multiple branches of the multiverse and their statistical properties follow from a single fundamental law of motion.  We model ``small-signal truncation'' in a somewhat stylized manner, but we argue that the precise details of its underlying physical mechanism do not affect the emergence of the Born rule.
\end{abstract}

\newpage
\tableofcontents
\newpage

\section{Introduction}

Quantum mechanics governs the behavior of matter and energy at atomic scales. Its modern textbook formulation (e.g.\ \citealt{griffiths2018introduction}) rests on three foundational postulates:
\begin{itemize}
    \item[(1)] The {\bf Schrödinger equation} governs the time evolution of the quantum state  in the absence of measurement (\citealt{schrodinger1926quantisierung,schrodinger1926verhaltnis}).

    \item[(2)] {\bf Wave function collapse} postulates the sudden reduction of the state vector during measurement, occurring when a quantum system interacts with the external world in an experiment (e.g.\ \citealt{VonNeumann1932}).

    \item[(3)] The {\bf Born rule}   determines the probability of obtaining a given measurement outcome (\citealt{born1926quantenmechanik}). This probability is proportional to the squared norm of the state vector after wave function collapse.

\end{itemize}
Despite its predictive success, quantum mechanics remains conceptually problematic. The issue is not its counterintuitive implications, but rather that a fundamental physical theory should not require three separate and partially incompatible laws of time evolution.

Physical theories traditionally follow a simple framework: initial conditions plus laws of motion. Classical mechanics, electrodynamics, and general relativity all share this structure. The multiple postulates of quantum mechanics deviate from this framework.
In particular, the wave function collapse postulate directly contradicts the Schrödinger equation. This ``measurement problem'' does not affect the ability of the theory to make coherent and successful predictions, but it exposes a fundamental inconsistency when our goal is a coherent description of nature.

Fortunately, Everett’s Many-Worlds Interpretation (\citealt{everett1957relative}), combined with quantum decoherence theory (e.g.\ \citealt{zeh1970interpretation}, \citealt{zurek1981pointer}), offers a logically consistent solution. This solution relies on two key principles:

\begin{itemize}
    \item[(a)] 
    {\bf One law of motion:}
    The Many-Worlds Interpretation eliminates wave function collapse as a fundamental postulate. Only the Schrödinger equation governs quantum evolution of the state vector.

    \item[(b)] 
    {\bf Macroscopic emergent phenomena:}
    Decoherence causes the quantum state to evolve into non-interacting branches, each corresponding to a different measurement outcome. Wave function collapse emerges as an apparent phenomenon within each branch (e.g.\ \citealt{joos1985emergence}).
    Each branch then continues to evolve as a separate ``world'' or ``universe''.
\end{itemize}
While significant progress has been made in understanding how the Schrödinger equation accounts for classical-like behavior in macroscopic systems, quantum decoherence remains an active area of research with some foundational questions still open. But for the purpose of our paper, we take the Many-Worlds Interpretation and the branching of the state vector into decoherent states as given.

The idea of a continuously branching multiverse may seem unconventional. This mirrors historical developments like the Theory of Evolution or the Big Bang theory, which were also counterintuitive when first proposed. But ultimately, the elegance and simplicity of the underlying theory tend to prevail over the seeming strangeness of the explanation. Crucially, the Many-Worlds Interpretation helps to restore the fundamental principle of a single, universal law of motion.
Although initially met with skepticism, it has since gained significant traction among a large portion of the physics community, with surveys of quantum physicists showing increasing acceptance of this interpretation (see e.g.\ \citealt{tegmark1998interpretation}; \citealt{schlosshauer2013snapshot}).

Nevertheless, even within the Many-Worlds Interpretation, a significant problem remains: The Schr\"odinger equation alone cannot generate all quantum mechanical predictions without incorporating the Born rule. While these postulates do not contradict each other, they violate the principle of a single fundamental law.

\paragraph{Example:}
Consider an experiment with $n=1000$ independent binary quantum measurements (such as Stern-Gerlach measurements), producing outcomes $y=(y_1,\ldots,y_n) \in \{0,1\}^n$. Assume that each measurement splits the state vector as $\psi_{y_i=1} + \psi_{y_i=0}$, with $\langle \psi_{y_i=1} \,|\, \psi_{y_i=0} \rangle =0$ and $\langle \psi_{y_i=1} \,|\, \psi_{y_i=1} \rangle / \langle \psi_{y_i=0}  \,|\, \psi_{y_i=0} \rangle = 1/4$ for all $i$.\footnote{We write $\langle \psi_1 \,|\, \psi_2\rangle$ for the inner product of any quantum states $\psi_1$ and $\psi_2$.
According to the Born rule, our assumtion
$\langle \psi_{y_i=1} \,|\, \psi_{y_i=1} \rangle / \langle \psi_{y_i=0}  \,|\, \psi_{y_i=0} \rangle = 1/4$
implies that the outcome $y_i=0$
 is four times as likely as the outcome $y_i=1$.
}
Then,
the Schrödinger equation with decoherence predicts $2^n$ distinct universes, one for each possible outcome sequence. The Born rule, however, makes strong additional predictions: it asserts that each individual measurement yields  $y_i=1$ with a probability of $20\%$.
Given the independence of measurements, this implies that the sum $\overline y = \sum_{i=1}^n y_i$ will be observed in the interval $[100,300]$ with virtual certainty (probability $1-2.2 \cdot 10^{-14}$). Remarkably, those outcomes
with $\overline y \in [100,300]$
represent less than $10^{-37}$ of all possible $2^n$ universes, demonstrating how the Born rule severely constrains which branches of the multiverse we observe.
\hfill $\square$

\bigskip

The Born rule's predictions have been experimentally verified to high precision. However, the status of the Born rule
as a fundamental law is problematic for two reasons: Firstly,
 it requires two independent laws of time evolution --- the Schrödinger equation for state evolution, and the Born rule to give the probability of observed multiverse branches (e.g.\ restricting us to a $10^{-37}$ fraction of possible universes with probability $1-2.2 \cdot 10^{-14}$ in the above example).
Secondly, it postulates a probability distribution over the decoherent branches of the multiverse, which are emergent rather than fundamental objects
--- a fundamental law of nature should not depend on emergent phenomena.

We propose resolving these issues analogously to the Many-Worlds solution of the measurement problem:

\begin{itemize}
    \item[(a)] 
     {\bf One law of motion:}
    We eliminate the Born rule as a fundamental postulate, retaining only an augmented  Schrödinger equation as a single law of motion.
    The augmentation we require is a mechanism for ``small-signal truncation'', which dynamically eliminated branches of the multiverse whose total amplitude becomes
    relatively low.
      
    \item[(b)] 
     {\bf Macroscopic emergent phenomena:}
    The Born rule emerges as a macroscopic consequence of this law, similar to how wave function collapse emerges from decoherence.
\end{itemize}

\medskip

\paragraph{Example continued:} Returning to the above example, assume all $n=1000$ binary measurements occur simultaneously. The state vector then decomposes as $\psi = \sum_{y \in \{0,1\}^n} \psi_y$, creating $2^n$ decoherent branches corresponding to components $\psi_y$. These states differ in their squared amplitudes $\phi_y := \langle \psi_{y} |\psi_{y} \rangle$. The  continuation of each state naturally continues to branch into many universes, that is, the total number of possible multiverse branches descending from the original state $\psi$  is much larger than $2^n$ after macroscopic timescales
In our framework, within this larger set of multiverse branches, virtually all have $\overline y$ within $[100,300]$ --- the fraction outside this interval is only $2.2 \cdot 10^{-14}$. This result arises due to ``small-signal truncation,'' which ensures that the number of descendant branches originating from each state $\psi_y$ is proportional to $\phi_y$ over macroscopic timescales. Consequently, instead of postulating measurement probabilities proportional to $\phi_y$ as in the Born rule, we find that the relative number of future multiverse branches naturally emerges in the correct proportion as a result of intrinsic branching dynamics and small signal truncation.
\hfill $\square$

\bigskip

This example illustrates the core idea of our approach. Rather than postulating probabilities for measurement outcomes, we demonstrate that branching dynamics and small-signal truncation together produce a distribution of future universe counts that is proportional to squared amplitudes. In the example, since almost all future branches lie within the expected outcome range ($\overline y \in [100,300]$), finding ourselves in one of these branches requires no probabilistic explanation.

\subsection{Related literature}

Various justifications for the Born rule have been proposed:

\begin{itemize}

\item \textbf{Hidden variables theories:} In Bohmian mechanics \citep{bohm1952suggested}, the Born rule arises from the distribution of hidden particle positions. Similar ideas appear in Valentini's work on subquantum equilibrium \citep{valentini1991signal}.

\item \textbf{Frequentist and infinite ensemble arguments:} Several authors have argued that the Born rule emerges from the law of large numbers applied to infinite sequences of measurements \citep{hartle1968quantum,farhi1989probability}. Related ideas are also developed within the Consistent Histories framework \citep{griffiths1984consistent,omnes1988logical}.\footnote{While both frequentist approaches and world-counting approaches involve counting procedures, they are conceptually distinct: frequentist arguments rely on repeated measurements of the same system, whereas world-counting arguments (such as in Hanson’s mangled worlds or our approach) count branches resulting from a single quantum event.}

\item \textbf{Decision-theoretic derivations:}   \cite{deutsch1999quantum} and  \cite{wallace2012emergent} proposed that rational agents must assign probabilities according to the Born rule, based on decision-theoretic axioms.

\item \textbf{Symmetry and envariance-based arguments:} 
\citet{zurek2005probabilities} introduced an approach where the Born rule follows from environment-assisted invariance (envariance), based on objective symmetries of entangled states rather than subjective ignorance. A more detailed and comprehensive derivation is presented in \citet[Chapter~3]{Zurek2025}. Related arguments based on self-location and symmetry were developed by \citet{carroll2014many,sebens2018self}.

\item \textbf{Self-location and measure-based approaches:}  \cite{vaidman1998schizophrenic,vaidman2011probability,vaidman2025probability} proposed that probabilities reflect self-location uncertainty, with weights given by the squared amplitudes. Further developments along these lines were made by  \cite{mcqueen2019defence} and by \citet{carroll2014many}.

\item \textbf{World counting:} 
\citet{hanson2003worlds,hanson2006drift} proposed that the Born rule emerges from counting worlds that survive a mangling threshold. 
\citet{strayhorn2008illustration} advocates for outcome counting as a natural probability measure in the Many-Worlds framework. 
\citet{sebens2015quantum} develops a related approach in the context of a no-collapse interpretation, where probabilities are determined by the density of worlds in configuration space.
Our analysis in this paper falls within this general category of world-counting approaches.

\item \textbf{Dynamical emergence through random perturbations:} 
\citet{landsman2017foundations} proposed that small random perturbations of the Schrödinger dynamics could explain the emergence of Born-rule statistics without postulating probabilities as fundamental. \citet{bonds2024quantum} developed this idea further using the method of arbitrary functions, showing that Born-rule probabilities arise as a universal feature in appropriate limits.

\item \textbf{Deterministic underlying theories:}   \citet{t2016cellular} proposes that quantum mechanics emerges from an underlying deterministic system governed by classical cellular automaton rules. The Born rule is argued to arise from epistemic uncertainty about which ontological microstate the system occupies, analogous to how statistical mechanics emerges from deterministic particle dynamics. Unlike Bohmian mechanics, 't Hooft's approach suggests that quantum states correspond to equivalence classes of ontological states, with probabilities reflecting typical distributions over these underlying states.

\end{itemize}

\bigskip

As reviewed by \citet{landsman2009born}, and still valid today, no generally accepted derivation of the Born rule has emerged, with concerns about circular reasoning or hidden assumptions persisting throughout the literature.

We also note that the decision-theoretic derivations, symmetry and envariance-based arguments, and self-location and measure-based approaches provide justifications for the Born rule, but do not offer dynamical derivations from underlying physical processes.
In this context, it is worth mentioning that Gleason's theorem \citep{gleason1957measures} shows that any probability measure satisfying natural mathematical conditions must take the Born rule form. However, such theorems do not explain the physical origin of the Born rule.

\bigskip

The work most closely related to ours is Hanson's ``Mangled Worlds'' approach (\citealt{hanson2003worlds,hanson2006drift}). Hanson also proposes that the Born rule emerges from the counting of worlds and provides a corresponding framework based on a truncation mechanism for world branches with small amplitudes. He justifies this truncation by arguing that interference between worlds of different amplitudes causes observers in smaller worlds to become ``mangled'', effectively introducing a lower cutoff in observable world sizes. When combined with the diffusion-like stochastic growth of world amplitudes, this produces an effective count of surviving worlds that matches the Born rule. The mathematical foundation of this derivation is very closely related to ours.

However, there are important differences between \cite{hanson2003worlds,hanson2006drift}
and our paper:
(i)~we propose that branches below the threshold are completely eliminated from the universal wavefunction, while Hanson maintains that such branches continue to exist but their observers become ``mangled'' through interference with larger worlds, losing their ability to perceive their original histories;
(ii)~our mechanism for justifying a threshold rule that guarantees the Born rule emerges from a self-consistent condition within the truncated branching process itself, while Hanson derives his threshold rule based on the position of the median in the untruncated measure distribution;
(iii)~we explicitly analyze the transition from discrete to continuous branching processes, while Hanson takes the continuous process approximation as given without studying the properties of the discrete branching process;
(iv)~our paper provides a more holistic view of the measurement process, explaining why typical quantum experiments operate far from the truncation threshold, thereby accounting for the consistent observation of Born rule statistics even for discrete branching processes (see Section~\ref{sec:unified}).
A more detailed comparison is provided in Appendix~\ref{app:Hanson}.

\bigskip

Regarding potential criticisms of our approach, \citet{vaidman1998schizophrenic,vaidman2025probability} argues that within the Many-Worlds interpretation, the Born rule must be introduced as a postulate connected to the measure of existence of branches, rather than derived from the dynamics. He points out that models involving truncation or pruning of low-amplitude branches, including Hanson's mangled worlds proposal and our small-signal truncation approach, may in principle lead to violations of relativistic causality by enabling superluminal signaling. This concern warrants further study.

Concerns have also been raised that world counts may depend sensitively on the choice of model and representation, which could complicate the interpretation of probabilities \citep{wallace2012emergent}. However, \cite{hanson2003worlds,hanson2006drift} argues that statistical predictions based on surviving world counts are robust to reasonable variations in the modeling details.

A further concern is raised by \citet{sebens2015killer} in the context of collapse models. He points out that if branches with low amplitude are first formed and only later eliminated, then agents might find themselves in worlds that are subsequently destroyed, making survival less likely and empirically disfavoring such theories. However, in our small-signal truncation model, branches are eliminated on extremely short timescales after diverging from surviving branches, implying that they have no meaningful time to come into existence as independent histories. In this sense, truncated branches never fully form, and the concern raised by Sebens does not apply.

\section{Conceptual framework}

This section outlines some conceptual foundations. We begin by summarizing how decoherence leads to the emergence of distinct quantum branches --- this is well-known, but we include it here to establish the necessary framework for our subsequent discussion.

\subsection{Quantum branching through decoherence}
\label{subsec:Branching}

The Schrödinger equation is a deterministic law of motion that governs how the universal state vector $\psi_{\tau}$ evolves over time $\tau$. For our purposes, this state vector describes the quantum state of the entire physical universe (or multiverse), or at least everything that can plausibly interact over the relevant time period  with a macroscopic observer of a given experiment (including the observer itself).
Due to the deterministic nature of the Schrödinger equation, if we know the state vector $\psi_{\tau_0}$ at any initial time $\tau_0$, the equation uniquely determines the state vector's value for all other times $\tau$.

Despite this underlying deterministic dynamics, decoherence theory explains that from the perspective of an observer inside the multiverse, the future is not unique. This is because over time the state vector branches into many decoherent components. Assume that at time $\tau_0$ the state vector $\psi_{\tau_0}$ describes one coherent (branch of the) universe which an internal observer perceives as their state of the world at time $\tau_0$. Then, at time $\tau_1 > \tau_0$ such that $\tau_1-\tau_0$ is large enough, the state vector will have additively split into multiple components $\psi_b(\tau_1)$ as follows:
\begin{align}
\psi(\tau_1) = \sum_{b \in {\cal B}} \psi_b(\tau_1) .
   \label{WaveFunctionBranching}
\end{align}
Here, ${\cal B}$ is a discrete set that labels the decoherent states. Each state $\psi_b(\tau_1)$, $b \in {\cal B}$, corresponds to a distinct classical reality that could emerge through the decoherence process. These branches are effectively independent, with negligible quantum interference between them, and each represents a different possible outcome or ``world'' that an observer might experience. The observer, being part of the quantum system, will find themselves in one of these branches, unable to directly detect the existence of the other branches due to decoherence.

It is remarkable that these multiple distinct futures emerge naturally from the strictly deterministic Schrödinger equation. This possibility was first suggested by 
\cite{everett1957relative}. The physical mechanism of this branching was later explained through the decoherence program, developed by \cite{zeh1970interpretation} and \cite{zurek1981pointer}, with important mathematical foundations established in \cite{joos1985emergence}.\footnote{
The mechanism of decoherence explains how classical reality emerges from the quantum world through the interaction between quantum systems and their environment. Rather than invoking collapse or external measurements, decoherence theory shows how environmental interactions naturally select certain preferred states --- known as ``pointer states" --- that remain stable under further interaction. 
}

For our analysis, we build upon these established decoherence results. The key insight is the wave function branching described by equation \eqref{WaveFunctionBranching}, where states $\psi_b(\tau_1)$ represent distinct, non-interfering universes emerging from a single initial state at time $\tau_0$. With elements of ${\cal B}$ labeled as $b_1, \ldots, b_{n}$, this evolution from unity to multiplicity is illustrated in Figure~\ref{fig:Branching}.

\begin{figure}[tb]
\begin{center} 
\begin{tikzpicture}
    \node[circle,draw] (psi0) at (0,0) {$\psi$};

    \node[circle,draw] (b1) at (6,2) {$b_1$};
    \node[circle,draw] (b2) at (6,0.7) {$b_2$};
    \node at (6,-0.5) {$\vdots$};
    \node[circle,draw] (bn) at (6,-2) {$b_{n}$};

    \draw (psi0) to[out=30,in=180] (b1);
    \draw (psi0) to[out=0,in=180] (b2);
    \draw (psi0) to[out=-30,in=180] (bn);

    \node at (0,-3.5) {time $\tau_0$};
    \node at (6,-3.5) {time $\tau_1$};
\end{tikzpicture}
 \end{center}
\caption{\label{fig:Branching} Branching of the state vector over time into decoherent states.}
 \end{figure}
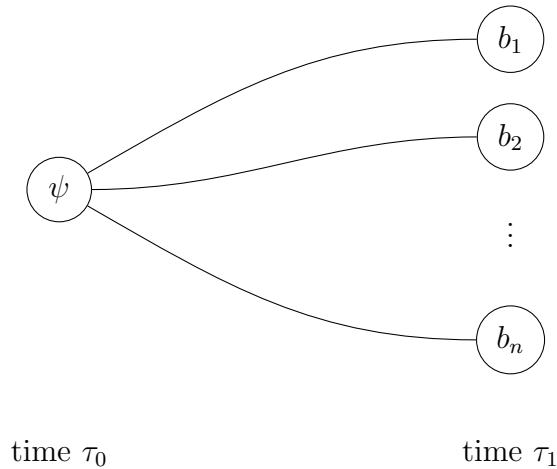

Furthermore, the squared norm of the various 
state vectors is important for our discussion below.
Decoherence theory ensures that the states
$\psi_b(\tau_1)$ are orthogonal, that is, 
we have $\left\langle \psi_a(\tau_1) \big| \psi_b(\tau_1) \right\rangle =0$ for all $a\neq b$ and $a,b \in {\cal B}$.
And since the Schrödinger equation conserves the squared norm over time, this implies that\footnote{
Remember that we write $\left\langle \psi_1 \big| \psi_2 \right\rangle$ for the inner product of any two state vectors $\psi_1$ and $\psi_2$. It is often 
customary in quantum mechanics to write the state vectors
themselves as $| \psi \rangle$, but we prefer the shorter notation $\psi$ in this paper.
}
\begin{align}
\left\langle \psi(\tau_0) \big| \psi(\tau_0) \right\rangle 
 = \left\langle \psi(\tau_1) \big| \psi(\tau_1) \right\rangle
 = \sum_{b \in {\cal B}} \left\langle \psi_b(\tau_1) \big| \psi_b(\tau_1) \right\rangle  .
   \label{WaveFunctionBranchingNorms}
\end{align}
The branching of the state vector shown in Figure~\ref{fig:Branching} occurs, for example, during a quantum measurement by a human observer. In that case, each branch $b \in {\cal B}$ corresponds to a different measurement outcome. But this kind of branching also happens naturally, without any observer. Decoherence occurs continuously as quantum systems interact with their environments. These interactions act like constant measurements, causing superpositions to evolve into stable, classical-like branches. No conscious observation is needed --- the environment itself is enough to induce the branching (see e.g.\ \citealt{zurek2003decoherence,schlosshauer2007quantum}).

\subsection{Replacing the probabilistic Born rule by universe counts}
\label{subsec:NoProbability}

Consider an experimental measurement with possible measurement outcomes labeled by $b \in {\cal B}$, and Figure~\ref{fig:Branching} illustrating the branching into distinct universes due to the measurement. Denote by $\phi_b := \left\langle \psi_b(\tau_1) \big| \psi_b(\tau_1) \right\rangle$ the squared norm of each branch. The Born rule gives the probability of observing measurement outcome $b \in {\cal B}$ as follows:
\begin{align}
\mathbb{P}(\text{observing outcome $b$}) = \frac{\phi_b}{\sum_{a \in {\cal B}} \phi_a} \; .
   \label{BornRule}
\end{align}
Note that we do not assume normalized state vectors in this paper, which explains the denominator $\sum_{a \in {\cal B}} \phi_a$ in the probability formula above --- that term is equal to one for normalized states.

Our claim is that the Born rule above is not required as a fundamental postulate of quantum mechanics. Indeed, our framework eliminates the need to postulate or derive any probability measure over experimental outcomes $b \in {\cal B}$.\footnote{
We replace the probability measure with a count measure over possible future universes, which can be interpreted probabilistically but doesn't have to be.
} The Schrödinger equation and decoherence fully determine how the universe branches into decoherent states $b \in {\cal B}$ with squared norms $\phi_b$ --- this is the complete physical description of what happens in an experiment. While the probabilistic interpretation in \eqref{BornRule} may be convenient, it is not necessary for explaining the observable implications of experiments.

While the Born rule provides important predictive power, as shown in our introductory example (where $\overline y \in [100,300]$ was predicted with near certainty), we must provide an alternative justification for such predictions without invoking probabilities. To develop this justification, we examine how the multiverse continues to branch naturally after the initial measurement.

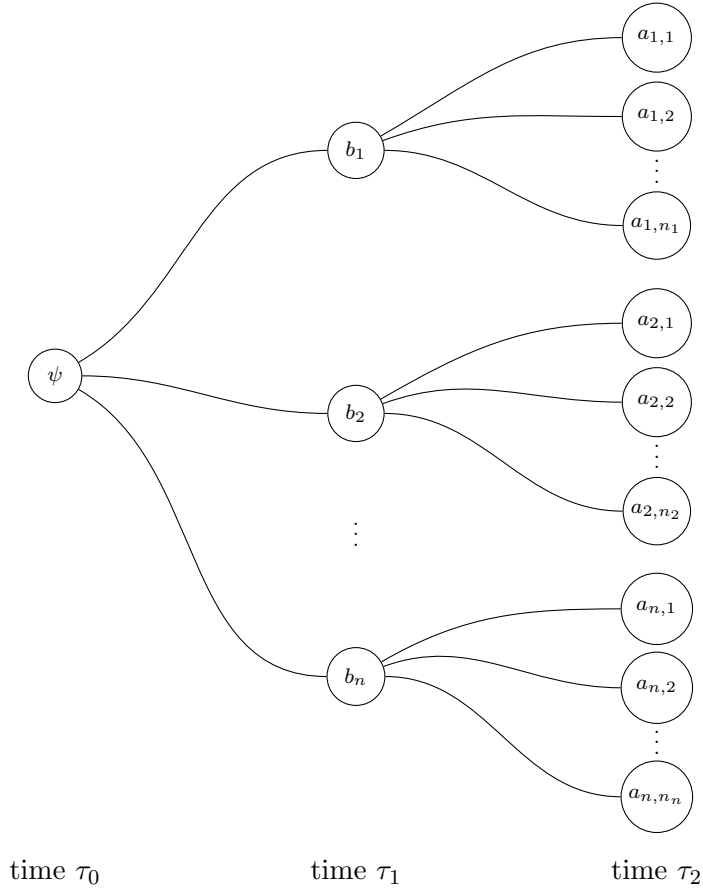
\begin{figure}[tb]
\begin{center} 
\begin{tikzpicture}[font=\scriptsize]
    \node[circle,draw] (psi0) at (0,0) {$\psi$};

    \node[circle,draw] (b1) at (4,3) {$b_1$};
    \node[circle,draw] (b2) at (4,-0.5) {$b_2$};
    \node at (4,-2) {$\vdots$};
    \node[circle,draw] (bn) at (4,-4) {$b_{n}$};

    \node[circle,draw] (a11) at (8,4.5) {$a_{1,1}$};
    \node[circle,draw] (a12) at (8,3.45) {$a_{1,2}$};
    \node at (8,2.83) {$\vdots$};
    \node[circle,draw] (a1n1) at (8,2.0) {\hspace{-0.1cm}$a_{1,n_1}\hspace{-0.1cm}$};

    \node[circle,draw] (a21) at (8,0.7) {$a_{2,1}$};
    \node[circle,draw] (a22) at (8,-0.35) {$a_{2,2}$};
    \node at (8,-0.97) {$\vdots$};
    \node[circle,draw] (a2n2) at (8,-1.8) {\hspace{-0.1cm}$a_{2,n_2}\hspace{-0.1cm}$};

    \node[circle,draw] (an1) at (8,-3.1) {$a_{n,1}$};
    \node[circle,draw] (an2) at (8,-4.15) {$a_{n,2}$};
    \node at (8,-4.77) {$\vdots$};
    \node[circle,draw] (ann) at (8,-5.6) {\hspace{-0.1cm}$a_{n,n_n}\hspace{-0.1cm}$};

    \draw (psi0) to[out=30,in=180] (b1);
    \draw (psi0) to[out=0,in=180] (b2);
    \draw (psi0) to[out=-30,in=180] (bn);
    
    \draw (b1) to[out=30,in=180] (a11);
    \draw (b1) to[out=20,in=180] (a12);
    \draw (b1) to[out=0,in=180] (a1n1);
    
    \draw (b2) to[out=30,in=180] (a21);
    \draw (b2) to[out=20,in=180] (a22);
    \draw (b2) to[out=0,in=180] (a2n2);
    
    \draw (bn) to[out=30,in=180] (an1);
    \draw (bn) to[out=20,in=180] (an2);
    \draw (bn) to[out=0,in=180] (ann);

    \node at (0,-6.6) {\small time $\tau_0$};
    \node at (4,-6.6) {\small time $\tau_1$};
    \node at (8,-6.6) {\small time $\tau_2$};
\end{tikzpicture}
 \end{center}
\caption{\label{fig:Branching2} Extended branching diagram showing two generations of decoherence: initial state $\psi$ first branches into states $b_1,\ldots,b_n$
upon experimental measurement
at time $\tau_1$, followed by each state $b_i$ branching 
naturally into $n_i$ states $a_{i,1},\ldots,a_{i,n_i}$ until time $\tau_2$.}
 \end{figure}

To understand how branching continues naturally after the measurement, consider Figure~\ref{fig:Branching2}. After the initial measurement creates states $b_i \in {\cal B}$ with squared norms $\phi_{b_i}$, each branch continues to evolve and split naturally through environmental decoherence. By time $\tau_2$, each state $b_i$ has branched into $n_i$ distinct states $a_{i,1},\ldots,a_{i,n_i}$.  

Our fundamental claim is that over macroscopic time scales, i.e.\ for $\tau_2 - \tau_1$ large, the ratio of branch counts for two measurement outcome states $b_i,b_j \in {\cal B}$
approaches the ratio of squared norms, that is,\footnote{In our formalizations of equation \eqref{OurBornRule} in Theorem~\ref{th:BasicModel} and \ref{th:GeneralizedModel} below, we introduce a small parameter $\epsilon$ that controls the signal truncation. The complete formulation then also involves taking the limit $\epsilon \rightarrow 0$, but we omit this technical detail here to focus on the core concepts. Also, the limit $\epsilon \rightarrow 0$ will actually not be necessary anymore in Theorem~\ref{th:ContinuousModelNoEpsilonLimit}.}
\begin{align}
  \lim_{\tau_2 \rightarrow \infty} \;
\frac{n_i} {n_j} = \frac{\phi_{b_i}}{\phi_{b_j}} \; .
   \label{OurBornRule}
\end{align}
Thus, while at time $\tau_1$ the measurement outcome states $b \in {\cal B}$ differ only in their squared norms $\phi_b$, these differences naturally translate into proportional differences in the number of future branches as the multiverse continues to evolve.

The Schrödinger equation alone cannot explain the relationship in \eqref{OurBornRule} due to its linearity: If $\psi(t)$ is a solution to the Schrödinger equation, then $c \, \psi(t)$ for any $c > 0$ is also a solution following identical dynamics. This scale invariance means the squared norm $\phi_b$ of a branch should have no effect on its future evolution, including the number of sub-branches it generates. Thus, to achieve the proportional relationship in \eqref{OurBornRule}, we must modify the Schrödinger equation, as discussed in the next subsection. The advantage of this modification is that \eqref{OurBornRule} emerges naturally from the system’s dynamics, rather than having to be postulated, as with the original Born rule.

Taking \eqref{OurBornRule} as given for the moment, we can use it to reproduce the experimental predictions of the traditional Born rule \eqref{BornRule} in two distinct ways.
The first way is pragmatic:
Since we cannot distinguish between the universes created at time $\tau_2$, it seems natural to assume that we find ourselves with equal probability in each of them. For macroscopic $\tau_2 - \tau_1$, this immediately implies that we observe outcome $b$ with probability proportional to $n_b$, which by \eqref{OurBornRule} equals $\phi_b/\sum_a \phi_a$. This recovers the Born rule \eqref{BornRule} and thus inherits all its predictions.

The second way to recover all testable predictions from \eqref{OurBornRule} avoids probabilistic concepts entirely. This means we cannot reproduce probabilistic statements like $\mathbb{P}(b=0) = \phi_0 / (\phi_0 + \phi_1)$ for a single Stern-Gerlach measurement with ${\cal B} = \{0,1\}$. But this is not a problem --- just as one cannot test a fair die with a single throw. All meaningful tests of the Born rule involve near-certain predictions, like $\overline y \in [100,300]$ in our example in the introduction. For such tests, \eqref{OurBornRule} is enough. At time $\tau_2$, almost all universes have $\overline y \in [100,300]$. Only a tiny fraction, $2.2 \cdot 10^{-14}$, violate this --- but they are so rare that predicting $\overline y \in [100,300]$ is natural even without using probability. If experiments regularly found $\overline y \notin [100,300]$, we would need a very strong explanation for why we keep finding ourselves in such rare branches.

\subsection{Small signal truncation}
\label{subsec:MotivationSmallSignalTruncation}

If we take the unmodified Schrödinger equation as our law of motion, we can never obtain \eqref{OurBornRule}, because the linearity of Schrödinger's equation ensures that the number of future branches of the state $\psi_b(\tau_1)$ remains independent of its squared norm $\phi_b = \left\langle \psi_b(\tau_1) \big| \psi_b(\tau_1) \right\rangle$.  
To establish a connection between branch counts and squared norms, we propose modifying the law of motion by introducing small-signal truncation --- a mechanism that eliminates components of the state vector $\psi(\tau)$ once their amplitude falls below a certain threshold. While this truncation may seem artificial, it could emerge naturally if quantum mechanics approximates an underlying discrete theory.

Consider an analogy with numerical simulation: A quantum simulation on a digital computer must approximate the continuous complex amplitudes of the Schrödinger equation using finite-precision arithmetic, causing numbers below machine precision to vanish. Similarly, if the fundamental state and dynamics of reality are discrete --- resembling a vast cellular automaton --- then the Schrödinger equation would be an emergent approximation, and some form of small-signal truncation could naturally arise from the system’s discrete nature. 

This analogy is not meant to imply that the universe is a simulation. Rather, it illustrates how any discrete theory underlying quantum mechanics would likely impose limits on amplitude precision, which could naturally give rise to small-signal truncation. Related ideas about discreteness leading to modifications of quantum dynamics have been proposed by \citet{buniy2006discreteness}.

Another possible justification for small-signal truncation is provided by Hanson's mangled worlds framework \citep{hanson2003worlds,hanson2006drift}, where interference between worlds of different amplitudes leads to the effective suppression of small branches. Since we do not commit to a specific physical mechanism for truncation, most of the results developed below also apply to the mangled worlds scenario.

To mathematically model small-signal truncation in this paper, we introduce a thresholding mechanism: when the amplitude of a state component $\psi_b(\tau)$ drops below a fixed threshold, it is removed from the total state vector of the multiverse. This truncation occurs continuously over time, while the system otherwise evolves according to the Schrödinger equation and undergoes natural branching due to decoherence. In the next section, we formalize this idea within a basic branching model and explore its implications.

\section{Basic branching model with truncation}
\label{sec:BasicModel}

In this section, we present a stylized model that demonstrates how quantum dynamics with small-signal truncation leads to the emergence of the Born rule. The model is deliberately simplified: it assumes that $K$-way branching occurs at fixed time intervals with constant branching ratios, and truncation is based on a threshold of the squared amplitude, with the time dependence of the threshold given exogenously (i.e.\ determined outside of our current modeling framework). We discuss generalizations of these assumptions in the next sections.

Time periods are labeled by $t=0,1,2,\ldots$, corresponding to physical times $\tau_0 < \tau_1 < \tau_2 < \ldots$, which are equidistant, that is, $\tau_t - \tau_{t-1}$ is constant. Following the decoherence mechanism discussed in Subsection~\ref{subsec:Branching}, 
starting from a single coherent state at $t=0$,
the wave function branches in each period $t=1,2,3,\ldots$. At each branching point, the squared norm of the state vector splits between the branches according to fixed ratios. For illustration, we show the case of three-way branching ($K=3$) in Figure~\ref{fig:KTree}.

\begin{figure}[tb]
\begin{center}
\begin{tikzpicture}[
    every node/.style={circle,draw,minimum size=0.5cm,inner sep=1pt,font=\scriptsize},
    grow=right
]
    \node (root) at (0,1) {$\;\;\psi\;\;$};

    \node (b11) at (4,2.4) {$b_1{=}1$};
    \node (b12) at (4,0) {$b_1{=}2$};
    \node (b13) at (4,-2.4) {$b_1{=}3$};
    
    \draw (root) to[out=45,in=180] (b11);
    \draw (root) to[out=0,in=180] (b12);
    \draw (root) to[out=-45,in=180] (b13);

    \node (b21) at (8,3.65) {$b_2{=}1$};
    \node (b22) at (8,2.7) {$b_2{=}2$};
    \node (b23) at (8,1.75) {$b_2{=}3$};
    \node (b24) at (8,0.55) {$b_2{=}1$};
    \node (b25) at (8,-0.4) {$b_2{=}2$};
    \node (b26) at (8,-1.35) {$b_2{=}3$};
    \node (b27) at (8,-2.45) {$b_2{=}1$};
    \node (b28) at (8,-3.4) {$b_2{=}2$};
    \node (b29) at (8,-4.35) {$b_2{=}3$};

    \draw (b11) to[out=45,in=180] (b21);
    \draw (b11) to[out=0,in=180] (b22);
    \draw (b11) to[out=-45,in=180] (b23);
    
    \draw (b12) to[out=45,in=180] (b24);
    \draw (b12) to[out=0,in=180] (b25);
    \draw (b12) to[out=-45,in=180] (b26);
    
    \draw (b13) to[out=45,in=180] (b27);
    \draw (b13) to[out=0,in=180] (b28);
    \draw (b13) to[out=-45,in=180] (b29);

    \node (bT1) at (14,4.8) {$b_t{=}1$};
    \node (bT2) at (14,3.8) {$b_t{=}2$};
    \node (bT3) at (14,2.8) {$b_t{=}3$};

    \node[draw=none] at (14,1.8) {$\vdots$};

    \node (bT4) at (14,0.5) {$b_t{=}1$};
    \node (bT5) at (14,-0.5) {$b_t{=}2$};
    \node (bT6) at (14,-1.5) {$b_t{=}3$};

    \node[draw=none] at (14,-2.5) {$\vdots$};

    \draw[dotted] (b21) .. controls +(right:2cm) and +(left:2cm) .. (bT1);
    \draw[dotted] (b21) .. controls +(right:2cm) and +(left:2cm) .. (bT2);
    \draw[dotted] (b21) .. controls +(right:2cm) and +(left:2cm) .. (bT3);
    \draw[dotted] (b21) .. controls +(right:2cm) and +(left:2cm) .. (13.5,2.1);
    \draw[dotted] (b21) .. controls +(right:2cm) and +(left:2cm) .. (13.5,1.5);

    \draw[dotted] (b22) .. controls +(right:2cm) and +(left:2cm) .. (bT4);
    \draw[dotted] (b22) .. controls +(right:2cm) and +(left:2cm) .. (bT5);
    \draw[dotted] (b22) .. controls +(right:2cm) and +(left:2cm) .. (bT6);
    \draw[dotted] (b22) .. controls +(right:2cm) and +(left:2cm) .. (13.5,-2.3);
    \draw[dotted] (b22) .. controls +(right:2cm) and +(left:2cm) .. (13.5,-2.9);

    \draw[dotted] (b23) .. controls +(right:1.5cm) and +(left:1.5cm) .. (11,-1.5);
    \draw[dotted] (b23) .. controls +(right:1.5cm) and +(left:1.5cm) .. (11,-1.75);
    \draw[dotted] (b24) .. controls +(right:1.5cm) and +(left:1.5cm) .. (11.5,-2.5);
    \draw[dotted] (b24) .. controls +(right:1.5cm) and +(left:1.5cm) .. (11.5,-2.75);
    \draw[dotted] (b25) .. controls +(right:1.5cm) and +(left:1.5cm) .. (12,-3.5);
    \draw[dotted] (b25) .. controls +(right:1.5cm) and +(left:1.5cm) .. (12,-3.75);
    \draw[dotted] (b26) .. controls +(right:1.5cm) and +(left:1.5cm) .. (12.5,-4.5);
    \draw[dotted] (b26) .. controls +(right:1.5cm) and +(left:1.5cm) .. (12.5,-4.75);
    \draw[dotted] (b27) .. controls +(right:1.5cm) and +(left:1.5cm) .. (10,-3.5);
    \draw[dotted] (b27) .. controls +(right:1.5cm) and +(left:1.5cm) .. (10,-3.75);
    \draw[dotted] (b28) .. controls +(right:1.5cm) and +(left:1.5cm) .. (10.5,-4.25);
    \draw[dotted] (b28) .. controls +(right:1.5cm) and +(left:1.5cm) .. (10.5,-4.5);
    \draw[dotted] (b29) .. controls +(right:1.5cm) and +(left:1.5cm) .. (11,-5.0);
    \draw[dotted] (b29) .. controls +(right:1.5cm) and +(left:1.5cm) .. (11,-5.25);

    \node[draw=none] at (0,-6) {$t=0$};
    \node[draw=none] at (4,-6) {$t=1$};
    \node[draw=none] at (8,-6) {$t=2$};
    \node[draw=none] at (11,-6) {$\cdots$};
    \node[draw=none] at (14,-6) {$t>2$};
\end{tikzpicture}
\end{center}
\caption{\label{fig:KTree} Three-way branching process evolving from a single state at $t=0$ to $3^t$ states in
period $t$. In each period $t \geq 1$, every state branches into three successor states labelled by $b_t \in \{1,2,3\}$, creating paths uniquely identified by the sequence $b=(b_1,\ldots,b_t)$.}
\end{figure}
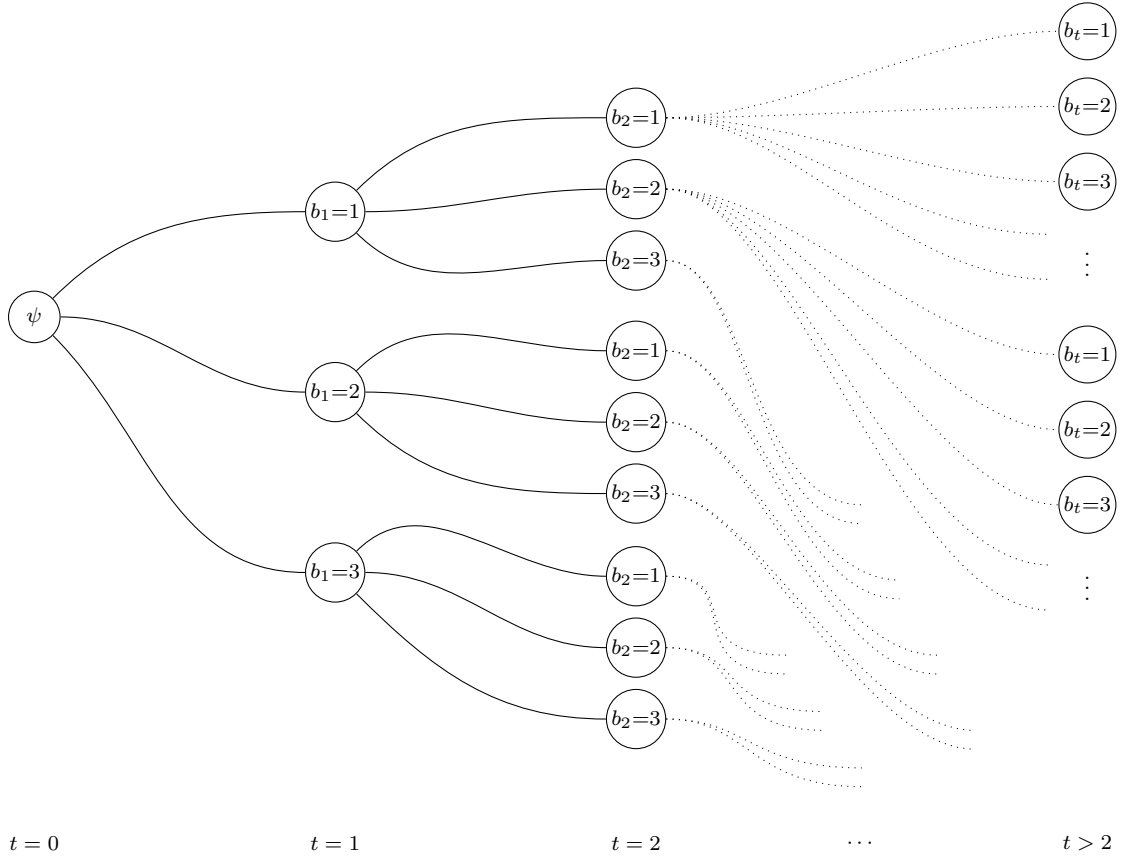

The $K^t$ decoherent states created by time $t$ can be uniquely labeled by $b=(b_1,\ldots,b_t) \in \{1,\ldots,K\}^t$. We denote the corresponding state vector as $\psi_t(b)$ with squared amplitude $\phi_t(b) = \langle \psi_t(b) | \psi_t(b) \rangle$. The squared amplitude of the initial state at $t=0$ is denoted $\phi_0 = \langle \psi_0 | \psi_0 \rangle$. While $\phi_t(b)$ depends on this initial value $\phi_0$, we leave this dependence implicit in our notation. As explained in Subsection~\ref{subsec:Branching}, in the absence of truncation, the total squared norm is preserved over time, that is,
$\phi_{t-1}(b_1,\ldots,b_{t-1}) = \sum_{k=1}^K \phi_t(b_1,\ldots,b_{t-1},k)$. In our basic model, we consider $K$-way branching with fixed ratios $\delta_k \in (0,1)$ for $k=1,\ldots,K$ where $\sum_{k=1}^K \delta_k = 1$. Thus, in the absence of truncation, our branching process would be given by
\begin{align}
   \phi_t(b) = \delta_{b_t} \, \phi_{t-1}(b_1,\ldots,b_{t-1}),
   \label{ProcessWithoutTheshold}
\end{align}
for all $t \in \{1,\ldots,T\}$ and $b \in \{1,\ldots,K\}^t$.
To introduce truncation into this model, we incorporate time-dependent threshold values $\xi_t \in (0,\infty)$. Our basic branching model with truncation is then defined as:
\begin{align}
   \phi_t(b) = 
   \begin{cases}
       \delta_{b_t} \, \phi_{t-1}(b_1,\ldots,b_{t-1}) \quad & \text{if } \delta_{b_t} \, \phi_{t-1}(b_1,\ldots,b_{t-1}) \geq \xi_t \,, \\
       0 & \text{otherwise.}
   \end{cases}
   \label{BasicModel}
\end{align}
This formulation implements small-signal truncation by eliminating branches whose squared amplitude falls below the threshold $\xi_t$ at time $t$. When a branch's squared amplitude becomes zero ($\phi_t(b)=0$), the corresponding state vector vanishes ($\psi_t(b)=0$), which we interpret as the termination of that branch of the process.  Through the threshold $\xi_t$, this model modifies the standard Schrödinger evolution by terminating branches whose squared amplitude falls below the threshold. 
For the basic model, we choose the threshold $\xi_t$ to be exogenously determined as
\begin{align}
   \xi_t = \epsilon \, \alpha^t,
   \label{ExogneousThreshold}
\end{align}
where $\alpha \in (0,1)$ and $\epsilon >0$ are fixed parameters. To motivate this functional form, consider how $\phi_t(b)$ would evolve in the absence of any threshold (as in equation \eqref{ProcessWithoutTheshold}). Obviously, 
$\phi_t(b)$ would decay exponentially with $t$, with decay rate determined by the sequence of branching ratios $\delta_{b_t}$ along the path $b$. Therefore, any positive threshold not decaying exponentially would eventually truncate all paths. The exponentially decaying threshold $\xi_t$ guarantees that non-trivial truncation behavior occurs for appropriate choice of $\alpha$.
In Section~\ref{subsec:BasicModelEndogenous}, we provide an endogenous rationale for both the form of $\xi_t$ and the value of $\alpha$ (i.e.\ both the functional form of $\xi_t$ and the value of $\alpha$ will be determined by the model itself). For now, however, we treat $\xi_t$ as exogenously specified
by \eqref{ExogneousThreshold}.

Equations \eqref{BasicModel} and \eqref{ExogneousThreshold} completely specify the
branching process with truncation.
Given this process, we can now count the number of surviving branches after $t$ time periods:
\begin{align}
   N_t(\phi_0) := \sum_{b \in \{1,\ldots,K\}^t} 
   \mathbbm{1}\left\{ \phi_t(b) >0  \right\} .
   \label{DefNT}
\end{align}
Note that while we suppress the dependence of $\phi_t(b)$ on the initial squared amplitude $\phi_0$ at $t=0$, we make this dependence explicit in $N_t(\phi_0)$.
Our version of the Born rule in equation \eqref{OurBornRule} can now be expressed as
\begin{align}
 \lim_{\substack{t \rightarrow \infty \\ \epsilon \rightarrow 0 }} \;
   \frac{N_t(\phi_{a})}
        {N_t(\phi_{b})}
     \; = \; \frac{\phi_{a}}   {\phi_{b}} ,
   \label{BasicBornRule}   
\end{align}
for any two initial squared amplitudes $\phi_a, \phi_b \in (0,\infty)$.
The limit taken here reflects an approximation where $\epsilon$ is very small
(otherwise we would constantly observe violations of the Schrödinger equation)
and $t$ becomes large (macroscopic time scales).
Equation \eqref{BasicBornRule} states our objective, but
the following theorem formally characterizes the limit of $N_t(\phi_{a})/N_t(\phi_{b})$
under appropriate conditions on $\delta_k$ and $\alpha$.

\begin{theorem}
   \label{th:BasicModel}
   Let $K \geq 3$, $\delta_k \in (0,1)$ with $\sum_{k=1}^K \delta_k = 1$,
   and $\alpha \in (0,1)$ be such that
   \begin{itemize}
      \item[(i)] 
      There exists $k,j,\ell \in {1,\ldots,K}$ such that
      $[\log(\delta_k / \delta_j)]\,\big/\,[\log(\delta_\ell/\delta_j)]$ is irrational.

      \item[(ii)] $\overline \delta   < \alpha < \max_k \delta_k$,
      where $\overline \delta := \left(\prod_{k=1}^K \delta_k\right)^{1/K} $.
   \end{itemize}
   For all $\epsilon>0$ and $\phi_0>0$,
   let $N_t(\phi_0)$ be as defined in equation \eqref{DefNT} for the process
   $\phi_t$ defined in equation \eqref{BasicModel} with threshold $\xi_t = \epsilon \, \alpha^t$.
   Then for all $\phi_{a},\phi_{b} \in (0,\infty)$, we have:
   \begin{align*}
      \lim_{0 \ll \epsilon^{-1} \ll t}
   \frac{N_t(\phi_{a})}
        {N_t(\phi_{b})}
      &= \left( \frac{\phi_{a}}   {\phi_{b}} \right)^\beta,
   & \text{with} \qquad
   \beta &:= \frac{  \log\left(\alpha \, \big/ \,\overline \delta \right)}{\frac 1 K \sum_{k=1}^K \left[ \log\left(\delta_k \, \big/ \,  \overline \delta\right) \right]^2}\;.
   \end{align*}
  Here, we write $\lim_{0 \ll \epsilon^{-1} \ll t}$ for the joint limit
   $t \rightarrow \infty$ and $\epsilon \rightarrow 0$ such that
   $\epsilon \, t \rightarrow \infty$.
\end{theorem}

The proof of Theorem~\ref{th:BasicModel} is provided in the appendix as a special case of the more general Theorem~\ref{th:GeneralizedModel} below. 
The functional form of the coefficient $\beta$ will also become much clearer after presenting Theorem~\ref{th:GeneralizedModel} below.
A justification of
the limit $0 \ll \epsilon^{-1} \ll t$ will be provided in Section~\ref{sec:unified}.

Assumption (i) requires that $[\log(\delta_k/\delta_j)]/[\log(\delta_\ell/\delta_j)]$ is irrational for some indices $k,j,\ell$. This ``non-lattice" condition prevents the path sums $\sum_{t=1}^T \log(\delta_{b_t})$ from being restricted to a lattice, which would cause survival patterns that depend on the discreteness of the lattice structure. This assumption explain why we need $K \geq 3$, since
for $K=2$ we always have that the path sums $\sum_{t=1}^T \log(\delta_{b_t})$ form a lattice, with corresponding discretized
survival patterns.

Condition (ii) requires $\overline{\delta} < \alpha < \max_k \delta_k$, where $\overline{\delta}$ is the geometric mean of the branching ratios. The lower bound involving the geometric mean arises because, in the absence of truncation, the expected logarithmic growth rate of a typical path is $\log\overline{\delta}$. Thus, $\overline{\delta} < \alpha$ ensures that the exponential decay of the threshold is slower than the typical decay of path amplitudes, creating negative drift and making truncation non-trivial. Without this condition, almost all paths would survive as $t \to \infty$. The upper bound guarantees some branches can grow relative to the threshold, preventing complete extinction of all paths as $t$ grows.

The key insight from Theorem~\ref{th:BasicModel} is that the Born rule emerges precisely when $\beta=1$, as this ensures the ratio $N_t(\phi_{a}) / N_t(\phi_{b})$ converges to the ratio $\phi_{a} / \phi_{b}$. To achieve $\beta=1$, we require:
\begin{align}
\alpha=\overline \delta \, \exp\left( \frac 1 K \sum_{k=1}^K    \left[ \log\left(\delta_k \, \big/ \,  \overline \delta\right) \right]^2  
\right) .
   \label{AlphaValueForBeta1}
\end{align}
This formula guarantees $\alpha > \overline \delta$ whenever $\delta_k$ varies across $k$, but does not necessarily ensure $\alpha < \max_k \delta_k$. The latter condition therefore imposes constraints on $(\delta_1,\ldots,\delta_K)$ for $\beta=1$ to be feasible.

For instance, with $K=3$ and $(\delta_1,\delta_2,\delta_3)=(1/6, 1/3, 1/2)$, equation \eqref{AlphaValueForBeta1} yields $\alpha=0.372041$. These parameters satisfy all assumptions in Theorem~\ref{th:BasicModel}, producing the desired conclusion with $\beta=1$.

By contrast, with $K=3$ and $(\delta_1,\delta_2,\delta_3)=(0.05, 0.45, 0.5)$, equation \eqref{AlphaValueForBeta1} yields $\alpha=0.691397$. This violates the condition $\alpha < \max_k \delta_k$, making it impossible to obtain the Born rule with these values of $\delta_k$.

It is difficult to provide a simple necessary and sufficient condition on $\delta_k$ that guarantees $\beta=1$ is feasible in Theorem~\ref{th:BasicModel}. However, for $K=3$, one useful sufficient condition is: If $\min(\delta_1, \delta_2, \delta_3) > (1 + 2e^{3/2})^{-1} \approx 0.10037$, then equation~\eqref{AlphaValueForBeta1} yields an $\alpha$ satisfying $\alpha < \max_k \delta_k$, ensuring $\beta=1$ is feasible. Note that condition (i) in Theorem~\ref{th:BasicModel} must still be satisfied independently. This simple example for a sufficient condition demonstrates that the Born rule emerges for a substantial range of branching parameters $\delta_k$, not just in isolated special cases.

The specific value of $\alpha$ in equation~\eqref{AlphaValueForBeta1} required for $\beta=1$ appears arbitrary at this stage. In Section~\ref{subsec:BasicModelEndogenous}, we show that the
threshold $\xi_t = \epsilon \, \alpha^t$
with that value for $\alpha$ emerges naturally when the threshold is determined endogenously within the model, providing a principled justification for the Born rule. Nevertheless, our discussion here has already demonstrated that the Born rule can emerge from a simple branching process with small-signal truncation, given appropriate conditions on the model parameters.

\section{More general models}
\label{sec:General}

The goal of the previous section was to present a simple deterministic branching model with small signal truncation that yields the Born rule (Theorem~\ref{th:BasicModel} with $\beta=1$), without invoking probabilistic concepts. In this section, we significantly generalize the branching process and introduce stochastic process representations in both discrete and continuous time, which offer various mathematical conveniences. 
However, the truncation threshold $\xi_t$ will remain exogenously specified as an exponential decay throughout this section. The extension to endogenously determined thresholds will be discussed in the next section.

\medskip

The process $\phi_t(b)$ defined in \eqref{BasicModel} is entirely deterministic, in line with the argument in Subsection~\ref{subsec:NoProbability} that all experimentally testable predictions of quantum mechanics can be derived without invoking probability. Nonetheless, even if probability is not part of our fundamental physical theory (as we argue),
we find that it remains a powerful and convenient mathematical tool for analyzing the branching process. This perspective is consistent with the discussion in \cite{spiegelhalter2024probability}.

We construct a probabilistic representation of our basic branching process $\phi_t(b)$ as follows. For $t=1,2,3,\ldots$, let $B_t$ be independent and identically distributed (i.i.d.) random variables uniformly distributed on $\{1,\ldots,K\}$, i.e., $\mathbb{P}(B_t = k) = 1/K$ for all $k$. Define
\begin{align*}
   \Phi_t &:= \phi_t(B_1,\ldots,B_t),
   &
   \Delta_t &:= \delta_{B_t} .
\end{align*}   
The stochastic process representation of our basic model branching model in \eqref{BasicModel}
is then given by
$$
    \Phi_t = 
   \begin{cases}
       \Delta_t \, \Phi_{t-1} \quad & \text{if }  \Delta_t \, \Phi_{t-1} \geq \xi_t \,, \\
       0 & \text{otherwise.}
   \end{cases}
$$
where the threshold $\xi_t$ remains exogenous and non-random.

Each realization of the stochastic process $(\Phi_t : t = 1,2,3,\ldots)$ corresponds to a path of the deterministic process $\phi_t(b)$. The stochastic process representation just replaces the uniform count measure over branching paths with a uniform probability measure. Since each sequence $(b_1,\ldots,b_t) \in \{1,\ldots,K\}^t$ occurs with equal probability $1/K^t$, we obtain
\begin{align*}
    \mathbb{P}\left(\Phi_t > 0 \, \big| \, \Phi_0 = \phi_0\right)
    =  \frac{N_t(\phi_0)} {K^t} \, ,
\end{align*}
where $N_t(\phi_0)$ denotes the number of untruncated paths through period $t$, as defined in \eqref{DefNT}. Consequently, for all $\phi_a, \phi_b > 0$,
$$
   \frac{N_t(\phi_{a})}
        {N_t(\phi_{b})}
     =
      \frac{ \mathbb{P}\left(\Phi_t > 0 \, \big| \,  \Phi_0 = \phi_a\right) }
    { \mathbb{P}\left(\Phi_t > 0 \, \big| \,  \Phi_0  =  \phi_b\right)  } \, .
$$
This equivalence means that instead of studying the limit of $N_t(\phi_a)/N_t(\phi_b)$, we can equally analyze the limit of the ratio of survival probabilities in the stochastic process representation. This is the goal of the next subsection.

\subsection{Discrete-time random walk with negative drift and  barrier}
\label{subsec:StochasticProcess}

Up to this point, we have simply rewritten our basic branching model as a stochastic process. The key generalization in this section lies in allowing for a more flexible specification of the random branching multiplier~$\Delta_t$ --- both in terms of its distribution (in particular, continuously distributed $\Delta_t$ is allowed for here) and the interpretation of its randomness (see the discussion in Subsection~\ref{subsec:discussion}). Despite this added generality, the structure of the branching model remains unchanged and is restated here for completeness:
\begin{align}
    \Phi_t = 
   \begin{cases}
       \Delta_t \, \Phi_{t-1} \quad & \text{if }  \Delta_t \, \Phi_{t-1} \geq \xi_t \,, \\
       0 & \text{otherwise,}
   \end{cases}
   \label{ModelStochastic}
\end{align}
with non-random initial state $\Phi_0 = \phi_0$. The random branching multipliers $\Delta_t > 0$ are assumed to be i.i.d.\ over~$t$ (dependence of $\Delta_t$ over $t$ could be incorporate and is briefly discussed in Subsection~\ref{subsec:discussion}). As in the previous section, the threshold follows the exogenous form $\xi_t = \epsilon \, \alpha^t$.
For interpretability as a branching process, one may impose $\Delta_t \leq 1$, but this assumption is not required for the mathematical results presented here.

To better connect with results from stochastic process theory, it is helpful to define the rescaled and logged process:
$$
  X_t := \log(\alpha^{-t} \, \Phi_t)\,,
$$
and introduce\footnote{
For the basic branching model in Section~\ref{sec:BasicModel}, we find
\begin{align}
  \mu &= \log\left(\alpha \, \big/\, \overline \delta \right) ,
  &
  \sigma &= \sqrt{\frac{1}{K}\sum_{k=1}^K\left[ \log\left(\delta_k \, \big/\, \overline \delta \right)   \right]^2 ,
}
  &
  U_t &= \frac 1 \sigma \log\left(\delta_{B_t} \, \big/\, \overline \delta \right) ,
  \label{BasicModelParametersShocks}
\end{align}  
but the more general definitions in \eqref{GeneralModelParametersShocks} apply throughout this section.
}
\begin{align}
    \mu &:= \log \alpha - \mathbb{E} \log(\Delta_t) ,
   &   
   \sigma^2 &:= {\rm Var}[ \log(\Delta_t) ] ,
   &
   U_t &:= \frac{\log(\Delta_t) -  \mathbb{E} \log(\Delta_t)}{\sigma} ,
  \label{GeneralModelParametersShocks}
\end{align}
where $\mu > 0$ and $\sigma > 0$ are constants, and $U_t$ is an i.i.d.\ standardized shock with $\mathbb{E}[U_t] = 0$ and $\mathbb{E}[U_t^2] = 1$ by construction.

Note that $X_t \in \mathbb{R} \cup \{-\infty\}$, where $X_t = -\infty$ indicates termination of the process. Using these definitions and $\xi_t = \epsilon \, \alpha^t$, we can rewrite \eqref{ModelStochastic} as
\begin{align}
X_t = \begin{cases}
X_{t-1} - \mu + \sigma \, U_t \quad & \text{if } X_{t-1} - \mu + \sigma \,  U_t \geq \log(\epsilon) \, , \\
- \infty & \text{otherwise.}
\end{cases}
    \label{SochasticProcessRepresentation}
\end{align} 
The initial value is $X_0 = \log(\phi_0)$.
This shows that $X_t$ is a random walk with negative drift and an absorbing barrier --- a class of processes well-studied in probability theory.
The following theorem generalizes Theorem~\ref{th:BasicModel} to this broader class of branching processes.

\begin{theorem}
   \label{th:GeneralizedModel}
   For $t=1,2,3,\ldots$,
   let $X_t$ be the stochastic process
   evolving according to \eqref{SochasticProcessRepresentation}
   for non-random $\mu,\sigma,\epsilon >0$, and 
   random shocks $U_t$ that are
   independent and identically distributed over~$t$ with
   $\mathbb{E} \, U_t = 0$, $\mathbb{E}U_t^2 =1$,
 $\mathbb{E}[|U_t|^{2+\gamma}] < \infty$ for some $\gamma > 0$,
and
$\mathbb{P}(U_t>\mu)>0$. 
In addition, we assume that the distribution of $U_t$ is non-lattice,
i.e., there do not exist constants $c \in \mathbb{R}$ and $h > 0$ such that 
$\mathbb{P}(U_t \in c + h\mathbb{Z}) = 1$.
We then have, for all $x_a,x_b \in \mathbb{R}$, that
   \begin{align*}
      \lim_{0 \ll \epsilon^{-1} \ll t}
    \frac{ \mathbb{P}\left(X_t \neq -\infty \, \big| \,  X_0 = x_a\right) }
    { \mathbb{P}\left(X_t \neq -\infty \, \big| \,  X_0 = x_b \right)  }
      &= \exp\left[ \frac{\mu}{\sigma^2}
           \left( x_a - x_b \right)
      \right] .
   \end{align*}
   Here, the limit notation $0 \ll \epsilon^{-1} \ll t$ means $\epsilon \to 0$, $t \to \infty$, and $\epsilon \, t \to \infty$.
\end{theorem}

The formal proof is provided in the appendix. A heuristic proof outline is provided at the end of this subsection, designed to offers some intuition for the result and assumptions.
Noting that
$$
   \frac{ \mathbb{P}\left(X_t \neq -\infty \, \big| \,  X_0 = x_a\right) }
    { \mathbb{P}\left(X_t \neq -\infty \, \big| \,  X_0 = x_b \right)  }
    = \frac{\mathbb{P}(\Phi_t > 0 \,|\, \Phi_0 = \phi_a)}
           {\mathbb{P}(\Phi_t > 0 \,|\, \Phi_0 = \phi_b)} ,
$$
with $\phi_{a/b} = \exp(x_{a/b})$,
and using $\Phi_t = \alpha^t \exp(X_t)$, the theorem can be restated in terms of the original process:
\begin{align}
      \lim_{0 \ll \epsilon^{-1} \ll t}
      \frac{\mathbb{P}(\Phi_t > 0 \,|\, \Phi_0 = \phi_a)}
           {\mathbb{P}(\Phi_t > 0 \,|\, \Phi_0 = \phi_b)}
      &= \left(\frac{\phi_a}{\phi_b}\right)^\beta,
      & \text{with} \qquad
      \beta := \frac{\mu}{\sigma^2}\;.
    \label{GeneralBeta}  
\end{align}
This result yields two important insights. Firstly, in the context of the basic model from the previous section, Theorem~\ref{th:GeneralizedModel} clarifies and strengthens Theorem~\ref{th:BasicModel}. Specifically, the somewhat obscure expression for $\beta$ in Theorem~\ref{th:BasicModel} simplifies elegantly to $\beta = \mu / \sigma^2$, where $\mu$ and $\sigma$ are given in \eqref{GeneralModelParametersShocks} --- or more specifically in \eqref{BasicModelParametersShocks}. The assumptions across both theorems correspond naturally --- for instance, the condition $\overline \delta < \alpha$ translates to $\mu > 0$ here.

Secondly, and more significantly, this stochastic process representation broadens the scope of the branching model considerably. It shows that the power-law relationship between initial conditions and survival probabilities holds across a wide class of processes characterized by negative drift and absorption. This generality is further explored in Subsection~\ref{subsec:discussion} below.

For future reference, it is helpful to express the coefficient $\beta$ in an alternative form. Notably, the parameter $\mu$ is not fundamental to the branching model, but instead decomposes as
$
    \mu = \log \alpha + \widetilde{\mu},
$
where $\log \alpha$ captures the rate of decay of the exogenous threshold $\xi_t = \epsilon \alpha^t$, and $\widetilde{\mu} = - \mathbb{E} \log(\Delta_t)$ is the drift of the logged process $\log \Phi_t$.
Substituting this decomposition into \eqref{GeneralBeta}, we obtain:
\begin{align}
    \beta = \frac{\widetilde{\mu} + \log \alpha}{\sigma^2}.
    \label{GeneralBeta2}
\end{align}
This form highlights the relationship between the power-law exponent $\beta$, the decay rate $\alpha$ of the threshold, and parameters $\widetilde \mu$ and $\sigma$ characterizing the branching process. Achieving $\beta = 1$ requires the threshold rate to satisfy
$
\log \alpha = \sigma^2 - \widetilde{\mu}
$,
which generalizes equation \eqref{AlphaValueForBeta1}.

\medskip

\begin{proof}[\bf Heuristic Proof of Theorem~\ref{th:GeneralizedModel}]
Consider the process $X_t$ that evolves according to equation \eqref{SochasticProcessRepresentation}. Without the absorbing barrier, $X_t$ would be a simple random walk:
\begin{align*}
X^*_t &=   X^*_{t-1} - \mu + \sigma \, U_t 
  = X_0 - \mu \, t + \sigma \sum_{s=1}^t  U_s \, .
\end{align*}
By the central limit theorem, under our assumptions on $U_t$, for large $t$, the sum $t^{-1/2} \sum_{s=1}^t  U_s$ converges to
$Z \sim {\cal N}(0,1)$, a normal distribution with mean 0 and variance $1$. Therefore, for large $t$, $X^*_t$ itself is approximately distributed 
as $X_0 - \mu \, t + \sigma \, t^{1/2} \, Z$, which is 
a normally distributed random variable with mean 
$X_0 - \mu  \, t$ and variance $\sigma^2 \, t$.

With this notation, the key steps in the proof of 
Theorem~\ref{th:GeneralizedModel} are as follows:
\begin{align*}
    \lim_{0 \ll \epsilon^{-1} \ll t}
    \frac{ \mathbb{P}\left(X_t \neq -\infty \, \big| \,  X_0 = x_a\right) }
    { \mathbb{P}\left(X_t \neq -\infty \, \big| \,  X_0 = x_b \right)  }
   &= 
    \lim_{0 \ll \epsilon^{-1} \ll t}
    \frac{  \mathbb{P}\left(X^*_s \geq \log(\epsilon) \text{ for all } s \in [0,t] \, \big| \, X_0 = x_a\right)  }
    { \mathbb{P}\left(X^*_s \geq \log(\epsilon) \text{ for all } s \in [0,t] \, \big| \, X_0 = x_b\right)    }
  \\
  &= 
   \lim_{0 \ll \epsilon^{-1} \ll t}
    \frac{  \mathbb{P}\left(X^*_t \geq \log(\epsilon)  \, \big| \, X_0 = x_a\right)  }
    { \mathbb{P}\left(X^*_t \geq \log(\epsilon) \, \big| \, X_0 = x_b\right)    }
  \\
  &=  \lim_{0 \ll \epsilon^{-1} \ll t}
 \frac{  \mathbb{P}\left( x_a - \mu \, t + \sigma \, t^{1/2} \, Z \geq \log(\epsilon) \right)  }
    {  \mathbb{P}\left( x_b - \mu \, t + \sigma \, t^{1/2} \, Z \geq \log(\epsilon) \right)    }
   \\ 
  &=  \lim_{t \rightarrow \infty}
 \frac{  \mathbb{P}\left( 
 Z \geq  - \frac{x_a}{\sigma} t^{-1/2} + \frac{\mu}{\sigma} \, t^{1/2} \right)  }
    { \mathbb{P}\left( 
 Z \geq  - \frac{x_b}{\sigma} t^{-1/2} + \frac{\mu}{\sigma} \, t^{1/2} \right)     }
  \\
  &=  \lim_{t \rightarrow \infty}
 \frac{ F_{Z}\left( \frac{x_a}{\sigma} t^{-1/2} - \frac{\mu}{\sigma} \, t^{1/2} \right) }
 { F_{Z} \left( \frac{x_b}{\sigma} t^{-1/2} - \frac{\mu}{\sigma} \, t^{1/2} \right) }
 \\
 &= \exp\left[ \frac{\mu}{\sigma^2}
           \left( x_a - x_b \right)
      \right]
\end{align*}
Here, the first equality just reformulates the survival probability in terms of the barrier-free process $X^*_t$, where non-termination means staying above $\log(\epsilon)$ at all times.

In the second equality we replace the condition of staying above the barrier at all times with the simpler condition of being above the barrier at the final time $t$. This is a non-trivial approximation that depends critically on the limit $0 \ll \epsilon^{-1} \ll t$. In that limit, paths that cross the barrier before time $t$ and then return above it at time $t$ become increasingly rare compared to paths that remain above the barrier throughout. This is because the negative drift makes recovery after crossing the barrier exponentially unlikely as $t$ increases, while $\epsilon \to 0$ ensures the barrier is far enough below typical paths. 

The third equality replaces $X^*_t$ by its large $t$ normal approximation introduced above.
The fourth equality rearranges the inequality to isolate $Z$ on the left side and drops the $\log(\epsilon)$
terms that are of smaller asymptotic order.
The fifth equality expresses the probabilities in terms of the standard normal cdf $F_{Z}$.
For the final equality, we evaluate the extreme tail
of the normal distribution using the asymptotic property $F_{Z}(-y) \approx  f_{Z}(-y)/y$ for large $y$ (where $f_Z$ is the standard normal density), and then evaluate the ratio.

We stress again that the arguments here are heuristic. However,
power laws (as in \eqref{GeneralBeta}) and exponential decays (after taking logs) in the tail behavior of stochastic processes are ubiquitous in both theoretical probability and statistical physics. That is, the result here is not surprising from the perspective of this literature.
\end{proof}

\subsection{Continuous-time random walk with negative drift and  barrier}
\label{subsec:ContinuousTimeModel}

We have already taken one big step towards model abstraction in Theorem~\ref{th:GeneralizedModel} by considering the stochastic process version of our originally non-probabilistic branching process. We now take a second step in abstraction by transitioning from discrete to continuous time. This is again motivated mathematical convenience --- for instance, in proving the stationary distribution results required later in Section~\ref{subsec:BasicModelEndogenous}.
 
Specifically, we replace discrete time $t \in \{0,1,2,\ldots\}$ with continuous time $\tau \in [0,\infty)$ and model the dynamics of $\Phi_\tau$ via the stochastic differential equation:
\begin{equation}
\begin{aligned}
    \phantom{a} && d \log \Phi_\tau &= -\widetilde \mu \, d\tau +  \sigma \, dW_\tau,  \qquad && \text{if } \Phi_\tau  \geq  \xi_\tau, && \phantom{a} \\
    && \Phi_\tau &= 0, && \text{otherwise,}
\end{aligned}
\label{ContinuousTimeProcessExogenous}
\end{equation}
where $\widetilde \mu, \sigma > 0$ are drift and volatility parameters, and $W_\tau$ denotes standard Brownian motion.\footnote{The differential notation $d \log \Phi_\tau$, $dW_\tau$, and $d\tau$ follows standard usage in stochastic calculus and denotes infinitesimal changes.} The expression $-\widetilde{\mu} \, d\tau + \sigma \, dW_\tau$ is the continuous-time analogue of the discrete shock term $\log \Delta_t$, and the parameters $\widetilde{\mu}$ and $\sigma^2$ correspond to $- \mathbb{E} \log(\Delta_t)$ and $\mathrm{Var}[\log(\Delta_t)]$, respectively.\footnote{In principle, the physical branching process $\Phi_\tau$ should be non-increasing in $\tau$, as branching can only reduce total squared amplitude. In discrete time, this can be enforced by requiring $\Delta_t \leq 1$. The continuous-time process in \eqref{ContinuousTimeProcessExogenous} does not impose this restriction, since Brownian motion has full support on the real line. However, the negative drift ensures that for $\tau_2 > \tau_1$, the probability of $\Phi_{\tau_2} > \Phi_{\tau_1}$ decays exponentially in $\tau_2 - \tau_1$. Thus, one should really view \eqref{ContinuousTimeProcessExogenous} as an approximation of the branching process over sufficiently long horizons. See also Section~\ref{subsec:discussion}.}
The process starts at $\Phi_0 = \phi_0$, and the threshold remains exponentially decaying, $\xi_\tau = \epsilon \,\alpha^\tau$.
Again, we define the rescaled and logged process as
$$
X_\tau := \log\left( \alpha^{-\tau} \Phi_\tau \right),
$$
and let $\mu = \log \alpha + \widetilde{\mu}$. Then, the continuous-time counterpart of \eqref{SochasticProcessRepresentation} becomes
\begin{equation}
\begin{aligned}
    \phantom{a} && dX_\tau &= -\mu \, d\tau +  \sigma \, dW_\tau,  \qquad && \text{if } X_\tau \geq \log(\epsilon), && \phantom{a} \\
    && X_\tau &= -\infty, && \text{otherwise.}
\end{aligned}
\label{ContinuousTimeProcess}
\end{equation}
The first line of the last display describes $X_\tau$ as a Brownian motion with drift for $X_\tau$ above the barrier $\log(\epsilon)$. The second line enforces truncation (i.e., $X_\tau = -\infty$) upon hitting the barrier, as in the discrete-time case.
We now present the continuous-time analogue of Theorem~\ref{th:GeneralizedModel}.

\begin{theorem}
\label{th:ContinuousModelNoEpsilonLimit}
Let $\{X_\tau: \tau \geq 0\}$ be the stochastic process defined 
in \eqref{ContinuousTimeProcess} with parameters $\mu, \sigma, \epsilon > 0$. Then, for all $x_a, x_b > \log \epsilon$, we have:
\begin{align*}
\lim_{\tau \rightarrow \infty} \frac{\mathbb{P}(X_\tau \neq -\infty \, | \, X_0 = x_a)}{\mathbb{P}(X_\tau \neq -\infty \, | \, X_0 = x_b)} = \exp\left[\frac{\mu}{\sigma^2}(x_a - x_b)\right].
\end{align*}
\end{theorem}
\medskip
The proof of the theorem, given in the appendix,  draws on spectral theory for diffusion processes with absorbing boundaries from \cite{karlin1981second}.
See also \cite{asmussen2003applied}.
In terms of the original process $\Phi_\tau = \alpha^\tau \exp(X_\tau)$, this result yields the continuous-time counterpart of \eqref{GeneralBeta}:
\begin{align}
      \lim_{\tau \rightarrow \infty}
      \frac{\mathbb{P}(\Phi_\tau > 0 \,|\, \Phi_0 = \phi_a)}
           {\mathbb{P}(\Phi_\tau > 0 \,|\, \Phi_0 = \phi_b)}
      &= \left(\frac{\phi_a}{\phi_b}\right)^\beta,
    \label{GeneralBetaCont}  
\end{align}
with $\beta = \mu / \sigma^2$ unchanged.

However, a notable simplification arises in the continuous-time setting: The result here holds for any fixed $\epsilon > 0$, without requiring the limit $\epsilon \to 0$. In contrast, the discrete-time result in Theorem~\ref{th:GeneralizedModel} required this limit due to the discontinuous, step-wise nature of the process. Brownian motion, with its continuous sample paths, eliminates this granularity.
Still, the fact that \eqref{GeneralBetaCont} holds for finite $\epsilon$ is somewhat surprising.
The following inductive argument is not a proof, but offers a heuristic rationale for why the result here holds at finite~$\epsilon$.

\subsubsection*{Heuristic justification for Theorem~\ref{th:ContinuousModelNoEpsilonLimit}} 

Suppose that for some fixed time $s > 0$ and constant $C > 0$ (independent of $x$), the following asymptotic relation holds, for all $x>\log \epsilon$,
\begin{align}
\lim_{\tau \to \infty}
\frac{\mathbb{P}(X_\tau \neq -\infty \mid X_s = x)}{\exp\left( -\frac{\mu^2(\tau-s)}{2\sigma^2} \right)} = C \, \exp\left(\frac{\mu x}{\sigma^2}\right).
   \label{eq:InductiveAnsatz}
\end{align}
Taking this as given, we obtain the conclusion of Theorem~\ref{th:ContinuousModelNoEpsilonLimit} at time $s$. We now argue that if \eqref{eq:InductiveAnsatz} holds for small $s > 0$, then it must also hold at time $s = 0$, by virtue of the dynamics of the process $X_\tau$.

By the Markov property, for any $\tau > s$ and $x > \log \epsilon$:
\begin{align}
\mathbb{P}(X_\tau \neq -\infty \mid X_0 = x) 
= \mathbb{E}\left[ \mathbb{P}(X_\tau \neq -\infty \mid X_s) \cdot \mathbbm{1}_{\{X_s > \log \epsilon\}} \,\middle|\, X_0 = x \right]. \label{eq:Markov}
\end{align}
Over the short interval $[0, s]$, the process evolves as
\[
X_s = x - \mu s + \sigma W_s, \qquad \text{where } W_s \sim \mathcal{N}(0, s),
\]
because for $x - \log \epsilon \gg \sqrt{\sigma^2 s}$ the probability of hitting the barrier before time $s$ is negligible. More precisely, this probability decays like $\exp\left( -\frac{(x - \log \epsilon)^2}{2\sigma^2 s} \right)$. Thus, for such $x$, we may approximate $\mathbbm{1}_{\{X_s > \log \epsilon\}} \approx 1$ with high probability.

Substituting \eqref{eq:InductiveAnsatz} into \eqref{eq:Markov}, and using the moment generating function $\mathbb{E}[\exp(yW_s)] = \exp(y^2 s / 2)$ with $y = \mu/\sigma$, we obtain:
\begin{align*}
\lim_{\tau \to \infty} \frac{\mathbb{P}(X_\tau \neq -\infty \mid X_0 = x)}{\exp\left(-\frac{\mu^2\tau}{2\sigma^2}\right)} 
&\approx C \, \exp\left(\frac{\mu^2 s}{2\sigma^2}\right) \cdot \mathbb{E}\left[ \exp\left( \frac{\mu}{\sigma^2}(x - \mu s + \sigma W_s) \right) \right] \\
&= C \, \exp\left(\frac{\mu x}{\sigma^2} - \frac{\mu^2 s}{\sigma^2} + \frac{\mu^2 s}{2\sigma^2} \right) \cdot \exp\left(\frac{\mu^2 s}{2\sigma^2} \right) \\
&= C \, \exp\left( \frac{\mu x}{\sigma^2} \right).
\end{align*}
This confirms that \eqref{eq:InductiveAnsatz} also holds at time $0$. 
The key insight is that the exponential form $\exp(\mu x / \sigma^2)$ is preserved under Brownian motion with drift, and that the moment generating function of the Gaussian increment exactly offsets the drift term.

\subsection{Further discussion and remarks}
\label{subsec:discussion}

The previous two subsections introduced our generalized branching models in both discrete and continuous time and established their key mathematical implications in 
Theorems~\ref{th:GeneralizedModel} and~\ref{th:ContinuousModelNoEpsilonLimit}. We now provide some additional context and clarification.

\subsubsection*{Remarks on the discrete time model}

For the discrete-time model in \eqref{ModelStochastic}, we assume that the random multipliers $\Delta_t$ are independent and identically distributed over $t$, which entails two assumptions:

\begin{itemize}
    \item[(i)] \textbf{Stationarity:}  
    The distribution of $\Delta_t$ is assumed to not depend on $t$. This is a natural assumption in a stationary physical environment. 
    
    \item[(ii)] \textbf{Independence:}  
    The process is defined over discrete time steps $t \in \{0,1,2,\ldots\}$ corresponding to equidistant points in physical time $\tau_t$. If the time interval $\tau_t - \tau_{t-1}$ is small, assuming independence between $\Delta_t$ and $\Delta_{t-1}$ may be unrealistic. However, by choosing this interval sufficiently large, the independence assumption becomes increasingly reasonable.
    This flexibility in time scale selection provides our main justification for treating the $\Delta_t$ as independent in our formal results above. However, as should be plausible from the heuristic proof of Theorem~\ref{th:GeneralizedModel}, the results remain valid under certain forms of weak dependence: As long as the sequence $(\log(\Delta_t)\,:\,t=1,2,3,\ldots)$ satisfies a central limit theorem (CLT) of the form
    \[
        \frac{1}{\sqrt{t}} \sum_{s=1}^t \left[\log(\Delta_s) - \mathbb{E} \log(\Delta_s)\right] 
        \;\Rightarrow\; \mathcal{N}(0, \sigma^2_*),
    \]
    for some asymptotic variance $\sigma^2_*$, then the key conclusions of the model still hold, one just needs to replace the $\sigma^2$ by $\sigma^2_*$ to account for autocorrelation in $\Delta_t$. We chose not to incorporate autocorrelation into our formal results, as doing so would have further complicated the presentation without affecting the main insights.
\end{itemize}

\subsubsection*{One concrete generalized discrete time branching model}

To illustrate that our framework can accommodate more general forms of randomness beyond the stylized model of Section~\ref{sec:BasicModel}, we now present a branching model in which the transition multipliers $\Delta_t$ are generated as pseudorandom numbers over time. This allows us to simulate a more interesting branching process with $K=2$ branches in each period, a value of $K$ that was not allowed in 
Section~\ref{sec:BasicModel} since it violated the non-lattice condition.

We rely on a well-known class of pseudorandom number generators based on modular arithmetic to produce a deterministic branching structure that well-approximates true randomness. Specifically, we use a linear congruential generator (LCG) of the form $ x_{n+1} = (a \cdot x_n) \bmod p $, where $ p $ is a large prime and $ a $ is a primitive root modulo $ p $. Normalizing by $ p $ yields values that approximate independent draws from the uniform distribution on $ [0,1] $. This method is simple and widely used, with standard choices including $ p = 2^{31} - 1 $, $ a = 16807 $ (\citealt{park1988random}), or for higher precision, $ p = 2^{61} - 1 $, $ a = 6364136223846793005 $ (\citealt{o2014pcg}).

Using such numbers for $p$ and $a$, we generalize the basic branching model in \eqref{BasicModel} with $K=2$ as follows:
\begin{align*}
   \phi_t(b) &= 
   \begin{cases}
       \delta_t(b) \; \phi_{t-1}(b_1,\ldots,b_{t-1}) \quad & \text{if } \delta_t(b) \, \phi_{t-1}(b_1,\ldots,b_{t-1}) \geq \xi_t \,, \\
       0 & \text{otherwise,}
   \end{cases}
   \\[8pt]
   \delta_t(b) &= \frac{c_t(b)} {p} \; ,
   \\[8pt]
   c_t(b) &=
      \begin{cases}
         \left[a \cdot c_{t-1}(b_1,\ldots,b_{t-1})\right] \bmod p &  \text{if } b_t=2   \,,
        \\
      p-1-\left\{\left[a \cdot c_{t-1}(b_1,\ldots,b_{t-1})\right] \bmod p\right\}  &  \text{if } b_t=1   \,,
   \end{cases}
\end{align*}
with initial values specified by $\phi_0>0$ and $c_0=1$.
The key difference to the basic model in \eqref{BasicModel}
 is that the fixed $\delta_k \in [0,1]$ branching ratios were replaced with 
branch-specific ratios
$\delta_t(b) \in [0,1]$, which are generated from
the pseudo-random LCG process $c_t(b) \in \{0,1,2,\ldots,p-1\}$.
We continue to consider  $\xi_t = \epsilon \, \alpha^t$, and we choose
$$
    \alpha = e^{-11/12} \approx  0.326 
$$
to guarantee $\beta=1$ (see discussion below).
The definition of $N_t(\phi_0) $ in  \eqref{DefNT} is unchanged here, just specialized to $K=2$, that is,
$$
   N_t(\phi_0) := \sum_{b \in \{1,2\}^t} 
   \mathbbm{1}\left\{ \phi_t(b) >0  \right\} .
$$
We claim that for the branching model just introduced we have
\begin{align}
  \lim_{0 \ll \epsilon^{-1} \ll t} \;
   \frac{N_t(\phi_{a})}
        {N_t(\phi_{b})}
     \; = \; \frac{\phi_{a}}   {\phi_{b}} \, .
     \label{BornRuleGeneralizedExample}
\end{align}
To justify this claim using Theorem~\ref{th:GeneralizedModel},
we again construct a probabilistic representation of the branching process. For $t=1,2,3,\ldots$, let $B_t$ be i.i.d.\ random variables with $\mathbb{P}(B_t = 1) = 1/2$. Define
 $\Phi_t = \phi_t(B_1,\ldots,B_t)$ as before and
\begin{align*}
   \Delta_t &:= \delta_t(B_1,\ldots,B_t) .
\end{align*}   
We then approximately have $ \Delta_t \sim \text{i.i.d.} \, U[0,1] $. This approximation is highly accurate for several reasons. Firstly, when $ p $ is large, the normalized values $ c_t(b)/p $ are densely and nearly uniformly distributed across the interval $[0,1]$. Secondly, our use of a primitive root $ a $ ensures that the LCG achieves full period $ p-1 $, cycling through all integers in $\{1,2,\ldots,p-1\}$ before repeating. Thirdly, the mapping used for $ b_t = 1 $ (reflecting the LCG value about the midpoint) diversifies branching trajectories across the tree but does not affect the approximation $ \Delta_t \sim \text{i.i.d.} \, U[0,1] $. Finally, foundational results from number theory show that such generators, when properly parameterized, pass standard statistical tests for uniformity and independence. As a result, over sufficiently long time horizons, the statistical behavior of our deterministic model becomes effectively indistinguishable from one governed by truly uniform random multipliers.

Under this approximation, we compute
$
\widetilde{\mu} = -\mathbb{E}[\log \Delta_t] = 1$ 
and
$\sigma^2 = \mathrm{Var}[\log \Delta_t] = \frac{1}{12}$.
Remember also our choice $\alpha = e^{-11/12}$.
Substituting into the formula for $\beta$ from Section~\ref{subsec:StochasticProcess}, we find that:
\[
\beta = \frac{\widetilde{\mu} + \log \alpha}{\sigma^2}
= 1.
\]
Equation~\eqref{BornRuleGeneralizedExample} then follows as a direct consequence of Theorem~\ref{th:GeneralizedModel},
assuming that the approximation 
$\Delta_t \sim \text{i.i.d.} \, U[0,1]$ is sufficiently accurate.

\medskip
Although the branching structure in this model is binary ($K=2$), the transition multipliers $\Delta_t$ still take on a wide range of values. This variability does not arise from any external randomness, but from the evolution of the internal state variable $c_t$, which differs across branches and influences the branching dynamics. Conceptually, this reflects the fact that $\phi_t$ is not the full state of the system, but merely one scalar function of the quantum state vector $\psi_t$, whose evolution is governed by the Schrödinger equation. The full quantum state vector $\psi_t$ is a high-dimensional object encoding all physical degrees of freedom relevant to the multiverse evolution. In our stylized model here, the variable $c_t$ serves as a simple proxy for part of this internal structure. It determines the precise branching ratios at each step and evolves recursively along each path. This added complexity of the internal state ensures that the branching dynamics unfold in a much richer way --- though still entirely within a deterministic framework. 

\subsubsection*{Remarks on the continuous time model}

The continuous-time model introduced in \eqref{ContinuousTimeProcess} can be interpreted as an approximation of the discrete-time stochastic process \eqref{SochasticProcessRepresentation}, especially when we consider the evolution of the logged and rescaled process $X_t$ over long time horizons. This connection is motivated by a classical insight from probability theory: under suitable conditions, the scaled sum of independent (or weakly dependent) random variables converges in distribution to Brownian motion --- a result formalized by the functional central limit theorem (Donsker’s invariance principle).

Recall that in the discrete-time model, above the barrier, the process $X_t$ evolves according to
\[
X_t = X_{t-1} - \mu + \sigma U_t,
\]
where the $U_t$ are i.i.d.\ standardized shocks. Over many time steps, the cumulative sum $\sum_{s=1}^t U_s$ becomes approximately normal due to the central limit theorem, and more precisely, the path of $X_t$ (linearly interpolated) converges in distribution to a Brownian motion with drift:
\[
X_t \approx X_0 - \mu t + \sigma W_t,
\]
where $W_t$ denotes standard Brownian motion. This approximation becomes increasingly accurate as $t$ grows, provided the process is not absorbed.

While the truncation mechanism introduces complications, we again argue that the specific implementation of the trunction --- whether discrete or continuous --- does not affect our main results. Consequently, the continuous-time model remains a valid approximation to the discrete-time process, even in the presence of truncation. The fact that both models yield the same limiting behavior in Theorems~\ref{th:GeneralizedModel} and~\ref{th:ContinuousModelNoEpsilonLimit} provides strong ex post justification for this approximation.

\subsubsection*{More general truncation mechanisms}

The truncation criterion adopted in this paper --- eliminating a branch when its squared amplitude $\phi_t(b)$ falls below a deterministic threshold $\xi_t$ --- is admittedly a stylized simplification. More broadly, the small-signal truncation mechanism should be viewed not as fundamental, but as an emergent feature of deeper underlying dynamics (e.g.\ a fundamentally discrete evolution, as previously suggested). In any case, we argue that the precise physical implementation of truncation is not essential for the emergence of the Born rule. 

For example, one natural extension is to allow the truncation threshold to vary randomly over time:
\begin{align}
X_t = \begin{cases}
X_{t-1} - \mu + \sigma U_t, & \text{if } X_{t-1} - \mu + \sigma U_t \geq \log(E_t), \\
-\infty, & \text{otherwise},
\end{cases}
\end{align}
where $E_t$ are i.i.d.\ positive random variables independent of $U_t$. If we parameterize the random barrier as $\log(E_t)=\log \epsilon + V_t$, where $\log \epsilon = \mathbb{E} \log(E_t)$, then we expect 
the asymptotic behavior in Theorem~\ref{th:GeneralizedModel} to remain unchanged,  because as $\epsilon \to 0$, the  fluctuations $V_t$ become negligible compared to the diverging distance between the initial condition and the average barrier $\log \epsilon$.

\section{Endogenously determined truncation threshold}
\label{subsec:BasicModelEndogenous}

\subsection{Main idea}
\label{subsec:MainIdeaEndogenous}

Let's revisit the basic branching model with truncation defined in equation \eqref{BasicModel}. In our analysis in
Section~\ref{sec:BasicModel}, two aspects appeared rather arbitrary: the choice of exogenous threshold $\xi_t = \epsilon \, \alpha^t$ and the parameter tuning required to ensure $\beta=1$ in Theorem~\ref{th:BasicModel}. We now address both issues by modifying the model such that the threshold $\xi_t$ becomes a function of the entire multiverse's branching process at time $t$.

An important distinction: In Section~\ref{sec:BasicModel}, we analyzed the branching process starting from a single coherent branch at time $ t = 0 $ (or by comparing two such branches with different initial amplitudes $ \phi_a $ and $ \phi_b $), treating $ t = 0 $ as the moment our experiment occurred. Now, we examine the branching process of the entire multiverse, which began much earlier, with the trees in Section~\ref{sec:BasicModel} representing only small subtrees. The underlying branching mechanism remains unchanged, and Figure~\ref{fig:KTree} still applies. However, the key difference is that we now consider $ \psi_0 $ at $ t = 0 $ as the initial state of the entire multiverse, while our experiment takes place much later, when the overall process has reached a ``steady state''.

Following our established notation, at time $t$, all possible branches (surviving and terminated) of the multiverse are labeled by states in the set $\{1,\ldots,K\}^t$, with $\phi_t(b)$ representing the squared amplitude of branch $b \in \{1,\ldots,K\}^t$. Let ${\cal B}_t^+ = \{ b \in \{1,\ldots,K\}^t \, : \, \phi_t(b) >0\}$ be the set of branches that have not been terminated by time $t$. We then propose the following endogenous threshold rule:
\begin{align}
    \xi_t = \frac {\varepsilon} {|{\cal B}_t^+|}
    \sum_{b \in {\cal B}_t^+} \phi_t(b) ,
    \label{EndogenousThreshold}
\end{align}
where $|{\cal B}_t^+|$ denotes the cardinality of ${\cal B}_t^+$.
Thus, we set $\xi_t$ equal to the average squared amplitude of all surviving branches multiplied by a small fixed number $\varepsilon>0$ (related to but distinct from the parameter $\epsilon$ used previously).

In stochastic process notation, where $\Phi_t = \phi_t(B)$ is a random variable, equation \eqref{EndogenousThreshold} can be written as:
\begin{align}
    \xi_t = \varepsilon \,
    \mathbb{E}\left[ \Phi_t \, \big| \, 
       \Phi_t>0 \right].
    \label{EndogenousThreshold2}
\end{align}
This threshold choice restores homogeneity (though not additivity) to our model dynamics: Multiplying the multiverse's state vector $\psi_t$ by any non-zero constant $c \in \mathbb{C}$ does not affect the dynamics, since both $\phi_t(b)$ and $\xi_t$ are multiplied by $|c|^2$. For our purposes, the constant $c$ can even depend on time $t$. Arguably, the threshold choice in \eqref{EndogenousThreshold} (or equivalently in \eqref{EndogenousThreshold2}) represents one of the simplest plausible models for $\xi_t$ that achieves this homogeneity.

The goal of this section is to show that once the process defined by the branching model  \eqref{BasicModel} with threshold rule \eqref{EndogenousThreshold} reaches its steady state, we obtain the Born rule result from Theorem~\ref{th:BasicModel} with exponent
\begin{align}
    \beta = \frac 1 {1-\varepsilon} \, .
    \label{BetaVsVarepsilon}
\end{align}
Our results so far were derived in the limit $\epsilon \rightarrow 0$, which now becomes $\varepsilon \rightarrow 0$, and we obtain $\beta=1$ in this limit as desired. Thus, combining our previous results with this endogenous threshold approach guarantees the Born rule in the limit of large $t$ and small $\varepsilon$.

The formal derivation of the claim just made will be done for the continuous-time version of the model that was introduced in Section~\ref{subsec:ContinuousTimeModel}. This is because of mathematical tractability rather than conceptual necessity.\footnote{The continuous-time approach also elegantly resolves the simultaneity issue present in equations \eqref{BasicModel} and \eqref{EndogenousThreshold}, where $\phi_t(b)$ and $\xi_t$ must be determined jointly. In the discrete-time case, a natural alternative would be to define $\xi_t = \frac {\varepsilon} {|{\cal S}_{t-1}^+|} \sum_{b \in {\cal S}_{t-1}^+} \phi_{t-1}(s)$ in terms of the state of the process in the previous period,
 though this slightly modifies the quantitative meaning of $\varepsilon$.} We consider the continuous-time process as a good approximation of its discrete-time counterpart.

We emphasize again that the threshold rule in \eqref{EndogenousThreshold} (or \eqref{EndogenousThreshold2}) should not be interpreted as a fundamental law of nature. It is intended as an emergent approximation that captures essential features of a deeper, more fundamental dynamics --- possibly one grounded in an underlying discrete dynamics, as discussed in Section~\ref{subsec:MotivationSmallSignalTruncation}.

\subsection{Continuous-time process with exogenous barrier}

We begin by reconsidering the exogenous threshold from our earlier analysis. Throughout this section, we work with the continuous-time version of the branching model introduced in Section~\ref{subsec:ContinuousTimeModel}. Our first goal is to characterize the stationary distribution of this process. Our second goal is to evaluate $\mathbb{E}[\Phi_\tau \mid \Phi_\tau > 0]$ in this steady state. This will serve as a useful precursor to the subsequent analysis of the endogenous threshold.

\subsubsection*{Stationary Distribution}

Consider the process $X_{\tau}$ defined in \eqref{ContinuousTimeProcess}.
We define the stationary (or steady state) distribution as the limiting distribution of $X_\tau$ conditional on survival up to time $\tau$, as $\tau \to \infty$. Formally, for any Borel set $A \subset [\log \epsilon, \infty)$,
\begin{align}
\pi_X(A) := \lim_{\tau \to \infty} \mathbb{P}(X_\tau  \in A \, | \, X_\tau \neq -\infty, \, X_0 = x_0),
  \label{LimitingDistribution}
\end{align}
which turns out to be independent of the initial value $x_0 > \log(\epsilon)$.
Let $f_X(x)$ denote the density of $\pi_X$ with respect to Lebesgue measure --- we then use the standard notation $\pi_X(dx) = f_X(x) \, dx$ in the following.

\begin{theorem}
\label{th:StationaryDistribution}
Consider the continuous-time process $X_\tau$ defined
in \eqref{ContinuousTimeProcess}
with fixed parameters $\mu,\sigma,\epsilon > 0$
and initial condition $x_0 > \log(\epsilon)$,
and let $\pi_X$ denote the limiting distribution
of the process
as defined in \eqref{LimitingDistribution}.
Then we have, for $x \geq \log(\epsilon)$,
\begin{align*}
\pi_X(dx) = \frac{\mu}{\sigma^2} \,
\exp\left\{-\frac{\mu}{\sigma^2} [x-\log(\epsilon)] \right\} \, dx.
\end{align*}
\end{theorem}
\medskip
\noindent
The proof of the theorem is given in the appendix. The continuous-time framework significantly simplifies this derivation by allowing us to leverage analytical tools from stochastic calculus, particularly the reflection principle for Brownian motion and martingale techniques. This permits direct characterization of the conditional process given survival, which would be more complicated in the discrete-time setting. 

Theorem~\ref{th:StationaryDistribution} shows that the stationary distribution of $X_{\tau}$, conditional on survival until time $\tau$, is exponential, shifted by the boundary $\log(\epsilon)$, with rate $\mu / \sigma^2$. In steady state, this means that most surviving paths lie only about $\sigma^2 / \mu$ above the boundary. This may seem at odds with our earlier argument that the process should be far from the boundary at the time of the experiment, but we will resolve this apparent contradiction in Section~\ref{sec:unified}.

\subsubsection*{Long-run behaviour of $\mathbb{E}[\Phi_\tau|\Phi_\tau>0]$}

The continuous-time version of the original branching process (before rescaling and taking logs) and of the exogenous threshold are given by $\Phi_\tau = \alpha^\tau \exp(X_\tau)$ and $\xi_\tau = \epsilon \alpha^\tau$, respectively. Using Theorem~\ref{th:StationaryDistribution}, we can now calculate the expectation of $\Phi_\tau/\xi_\tau$ conditional on survival in the steady state:
\begin{align*}
   \lim_{\tau \rightarrow \infty} \frac{   \mathbb{E}[\Phi_\tau|\Phi_\tau>0] }
   {\xi_\tau}
   &=
    \lim_{\tau \rightarrow \infty} \frac{   \mathbb{E}[\Phi_\tau|\Phi_\tau>0] }
   {\epsilon \, \alpha^\tau}
   = \frac 1 {\epsilon}  \, \lim_{\tau \rightarrow \infty}
     \mathbb{E}[\exp(X_\tau)|X_\tau \neq \infty] 
   \\
   &= \frac 1 {\epsilon} \int_{\log \epsilon}^\infty \, \exp(x) \, \pi_X(dx)
  =   \frac{\mu}{\sigma^2} 
  \int_{\log \epsilon}^\infty  
\exp\left\{\left(1-\frac{\mu}{\sigma^2}\right) [x-\log(\epsilon)] \right\} \, dx 
\\
&= \frac{\mu}{\mu-\sigma^2} = \frac{\beta}{\beta -1} = 
 \frac 1 \varepsilon \; ,
\end{align*}
where we assume $\beta := \mu/\sigma^2 > 1$ (otherwise the expectation does not exist), and in the final step we use \eqref{BetaVsVarepsilon} to express $\beta$ in terms of $\varepsilon$.

The last display shows that our branching process with exogenous threshold rule $\xi_\tau = \epsilon \alpha^\tau$  satisfies, in the steady state (as $\tau \rightarrow \infty$), the relationship 
$
   \xi_\tau = \varepsilon \,
    \mathbb{E}\left[ \Phi_\tau \, \big| \, 
       \Phi_\tau>0 \right]
$, which coincides exactly with our endogenous threshold rule in \eqref{EndogenousThreshold2}. While this does not constitute a full proof of our claim in Section~\ref{subsec:MainIdeaEndogenous} above --- since the model remains defined by an exogenous threshold --- it serves as a crucial consistency check.

The calculation also reveals what makes $\beta=1$ special: it is the critical point at which $\mathbb{E}[\Phi_\tau|\Phi_\tau>0]$ ceases to be well-defined in the steady state. This explains why our endogenous threshold rule, formulated in terms of $\mathbb{E}[\Phi_\tau|\Phi_\tau>0]$, naturally singles out $\beta=1$ as the special case in the limit $\varepsilon \rightarrow 0$, though we still have to formally derive this result for the actual model with endogenous threshold in the next subsection.

\subsection{Continuous-time process with endogenous barrier}

Let's now consider the model with endogenous threshold, which we again formulate in continuous time and in
close analogy to \eqref{ContinuousTimeProcessExogenous} above.
For constants  $\widetilde \mu \in \mathbb{R}$, $\sigma > 0$, and $\varepsilon \in (0,1)$, we consider the process\footnote{
For the results in this section, the sign of $\widetilde{\mu}$ is inconsequential. In fact, through rescaling, we could normalize the process to satisfy $\widetilde{\mu} = 0$ and $\sigma = 1$ without any loss of generality from a purely mathematical perspective. However, in the context of a physical branching process, as described in Sections~\ref{sec:BasicModel} and~\ref{sec:General}, we have $\widetilde{\mu} > 0$, meaning that the process $\log \Phi_\tau$ exhibits a negative drift.
}
\begin{equation}
\begin{aligned}
    \phantom{a} && d \log \Phi_\tau &= -\widetilde \mu \, d\tau +  \sigma \, dW_\tau,  \qquad && \text{if } \Phi_\tau  \geq  \varepsilon \,
    \mathbb{E}\left[ \Phi_\tau \, \big| \, 
       \Phi_\tau>0 \right], && \phantom{a} \\
    && \Phi_\tau &= 0, && \text{otherwise,}
\end{aligned}
\label{ContinuousTimeProcessEndogenous}
\end{equation}
where again $W_\tau$ is a standard Brownian motion.
The goal of this subsection is to demonstrate that the process in \eqref{ContinuousTimeProcessEndogenous} leads to a long-term behavior of the endogenous threshold,  
$\varepsilon \, \mathbb{E}\left[ \Phi_\tau \mid \Phi_\tau > 0 \right]$, that is 
equivalent to our previous exogenous threshold $\epsilon \, \alpha^\tau$
with parameter $\alpha$ such that we obtain $\beta = 1$ for small $\varepsilon$. The following theorem formalizes that result.

 \begin{theorem}
\label{th:EndogenousThresholdAsymptotics}
Consider the continuous-time process $\Phi_\tau$ defined in \eqref{ContinuousTimeProcessEndogenous} with parameters $\widetilde \mu \in \mathbb{R}$, $\sigma > 0$ and $\varepsilon \in (0,1)$.
Assume that at time $\tau = 0$ the process starts at a non-random initial value $\Phi_0 = \phi_0 > 0$, let 
$\xi_\tau := \varepsilon \, \mathbb{E}[\Phi_\tau \, | \, \Phi_\tau > 0]$ be the 
value of the threshold implied by the time-evolution of the process,
and let $c_0 := \phi_0 \, \varepsilon / (1-\varepsilon)$
and $\alpha := \exp\left( \frac{\sigma^2}{1 - \varepsilon} - \widetilde{\mu} \right)$.
Then, we have:
$$
    \lim_{\tau \to \infty} \frac{\log(\xi_\tau/c_0)}{\tau} = \log \alpha .
$$
\end{theorem}

\medskip
\noindent
We previously introduced $ \xi_\tau $ and $ \alpha $ differently, but within the endogenous threshold model here, we now define them as in Theorem~\ref{th:EndogenousThresholdAsymptotics},
in terms of this process and its parameters.
 The proof is provided in the appendix, but given the discussion in the previous subsection, the result is unsurprising: In the long run, the threshold follows  
$$
   \xi_\tau \approx c_0 \, \alpha^\tau,
$$
meaning the endogenous threshold behaves precisely like the exogenous threshold assumed in earlier sections. Moreover, substituting the expression for $ \alpha $ from Theorem~\ref{th:EndogenousThresholdAsymptotics} into \eqref{GeneralBeta2} yields  
$\beta = 1/(1-\varepsilon)$
which we already anticipated in equation \eqref{BetaVsVarepsilon}. Consequently, in the limit $ \varepsilon \to 0 $, we obtain $ \beta = 1 $.

This confirms the objective of the current section.
It is important to reiterate that, while we continue to use the same branching process as before --- now with an endogenized threshold rule --- the interpretation has changed: we now assume the process began long ago, reaching a steady state by the time any quantum experiment is performed. From this perspective, the resulting threshold behavior $\xi_\tau$ appears exogenous when viewed within the smaller subtree of the multiverse considered in earlier sections. The next section will further clarify this overall picture.

\section{A comprehensive view of measurement and Born rule} 
\label{sec:unified}

The previous sections introduced the key components and results of our framework step by step. With the necessary technical tools now in place, we are ready to present a more comprehensive mathematical description of the measurement process and the emergence of the Born rule. 
In doing so, this section will also provide a conceptual justification for the limiting conditions $\epsilon \to 0$ and $\epsilon \cdot t \to \infty$ that appeared in Theorems~\ref{th:BasicModel} and~\ref{th:GeneralizedModel}. 

For mathematical convenience, we again work with the continuous-time branching process introduced in Section~\ref{subsec:ContinuousTimeModel}. More specifically, the setup we analyze in this section is as follows:

\begin{itemize}
    \item[(A)] We assume that the process $X^{\rm past}_{\tau}$ has evolved according to \eqref{ContinuousTimeProcess} for a long time in the past,  $\tau < 0$. We also assume that the parameters $\mu$ and $\sigma$ satisfy $\mu/\sigma^2 = 1$. As shown in Theorem~\ref{th:StationaryDistribution}, this implies that the distribution of $X^{\rm past}_0$ at time $\tau = 0$ is the stationary exponential distribution with density
    \[
        f_X(x) = \frac{1}{\sigma^2} \, \exp\left(-\frac{x - \log(\epsilon)}{\sigma^2}\right), \quad x \geq \log(\epsilon).
    \]

    \item[(B)] At time $\tau = 0$, a quantum measurement occurs, causing a discrete branching event:
$$
        X_0 = \log(\Delta) + X^{\rm past}_0,
$$
    where $\Delta = \delta_B$, and $B$ is a discrete random variable uniformly distributed over $\{1,\ldots,K\}$, indexing the $K$ possible measurement outcomes. Each outcome $k$ is associated with a branching ratio $\delta_k \in (0,1)$, with  $\sum_{k=1}^K \delta_k = 1$. In terms of the original  process $\Phi_\tau = \alpha^\tau \exp(X_\tau)$, this branching step corresponds to $\Phi_0 = \Delta \cdot \Phi_0^{\text{past}}$.

    \item[(C)] After the measurement, the process $X_\tau$ continues to evolve forward in time for $\tau > 0$ according to \eqref{ContinuousTimeProcess}, using the same parameters $\mu$ and $\sigma$ (so that $\mu/\sigma^2 = 1$). We assume that the pre-measurement process $(X^{\rm past}_\tau \,:\, \tau < 0)$, the measurement outcome $B$, and the post-measurement process $(X_\tau \,:\, \tau > 0)$ are all mutually independent.
\end{itemize}
\medskip
\noindent
Under this setup, we now examine the distribution of measurement outcomes 
conditional on survival until time $\tau$, as $\tau \rightarrow \infty$.

\begin{corollary}
\label{cor:Born}
Let the setup be as described in (A), (B), and (C), with fixed branching weights $\delta_k \in (0,1)$. Then:
\[
\lim_{\tau \to \infty} \mathbb{P}(B = k \mid X_\tau \neq -\infty) = \delta_k.
\]
\end{corollary}
\medskip
\noindent
Given our earlier discussion and results, this corollary should be unsurprising.
However, there is a subtle shift in perspective: In Theorem~\ref{th:ContinuousModelNoEpsilonLimit} we computed the large $\tau$ limit of the survival probability
$\mathbb{P}(X_\tau \neq -\infty \mid X_0 = x_0)$, which, in the context of the measurement setup here, corresponds to
$\mathbb{P}(X_\tau \neq -\infty \mid B = k)$, since each outcome $k$ determines the initial value $X_0 = \log(\delta_k) + X_0^{\text{past}}$.
Here, by contrast, we are interested in the probability of a particular measurement outcome $B = k$ conditional on survival.
To relate these two perspectives, we apply Bayes' rule:
\[
\mathbb{P}(B = k \mid X_\tau \neq -\infty)
= \frac{\mathbb{P}(X_\tau \neq -\infty \mid B = k) \cdot \mathbb{P}(B = k)}
{\sum_{j=1}^K \mathbb{P}(X_\tau \neq -\infty \mid B = j) \cdot \mathbb{P}(B = j)}.
\]
From our previous results, we know that $\mathbb{P}(X_\tau \neq -\infty \mid B = k) \sim C \cdot \delta_k$ for large $\tau$, where $C$ is a constant that does not depend on $k$.
Since $\mathbb{P}(B = k) = 1/K$ for all $k$, the denominator is also proportional to $\sum_k \delta_k = 1$, and we obtain the corollary.

At first glance, the probabilistic result in Corollary~\ref{cor:Born} may seem at odds with our argument in Subsection~\ref{subsec:NoProbability},
where we rejected the need to postulate probabilities over measurement outcomes.
However, in that earlier discussion we also introduced the idea of a ``pragmatic'' interpretation --- assigning equal subjective probability across decoherent branches.
This pragmatic stance is precisely what we adopt here, and it leads directly to the Born rule as a conditional probability statement in Corollary~\ref{cor:Born}.

\bigskip

Next, we analyze the asymptotic behavior of the conditional median of \( X_0 \) as \( \tau \to \infty \). Before considering the full measurement setup described by (A), (B), (C), it is helpful to  consider a simpler result for the continuous-time process defined in \eqref{ContinuousTimeProcess}
on its own. Assuming that this process has reached its steady state distribution in Theorem~\ref{th:StationaryDistribution} by time \( \tau = 0 \), one finds that, as $\tau \to \infty$,
\begin{align}
\operatorname{Med}(X_0 \mid X_\tau \neq -\infty)
= \log(\epsilon) + \log(2) + \frac{\sigma^2}{\mu} \log(\tau) + o(1),
   \label{ConditionalMedian}
\end{align}
where the remainder term $ o(1) $ denotes a function that converges to zero as $\tau \rightarrow \infty$.
The result in \eqref{ConditionalMedian} follows from the observation that, conditional on survival until $\tau$, the distribution of \( X_0 \) converges to a shifted Gamma distribution whose scale parameter grows like \( \log(\tau) \). In fact, an analogous asymptotic expression holds for any other conditional quantile of \( X_0 \); the term \( \log(2) \) in the leading expression would simply be replaced by a different constant depending on the quantile level.

The result in \eqref{ConditionalMedian} is significant for our purposes because it reveals that, conditional on survival until large \( \tau \), we obtain not only the Born rule as stated in Corollary~\ref{cor:Born}, but also the additional insight that the typical past value of the difference \(X_0 - \log \epsilon \) increases proportionally to \( \log(\tau) \). In other words, after macroscopic times \( \tau \), the conditional value of \( X_0 \) should be expected to be well above the threshold \( \log \epsilon \). This explains why violations of the Schrödinger equation due to small-signal truncation are not constantly observed in quantum experiments.

\bigskip

However, the conclusion in the last sentence ignores an important point: the measurement step (B) itself reduces the value of \( X_0 \), particularly when \( \delta_k \) is small. To address this, the following theorem considers the full measurement setup involving steps (A), (B), and (C), and explicitly accounts for the possibility that \( \delta_k \) may be very small. Mathematically, we capture this by allowing \( \delta_k \to 0 \) as \( \tau \to \infty \). To be clear, we are not suggesting that \( \delta_k \) physically varies with time; rather, this formulation enables us to analyze asymptotic regimes in which both \( \tau \) is large and \( \delta_k \) is small.

\begin{theorem}
\label{thm:asymptotic_median}
Consider the setup as described in (A), (B), and (C), and fix a measurement outcome $k \in \{1,\ldots,K\}$.
Suppose that either \( \delta_k \) is fixed as \( \tau \to \infty \), or that 
\( \delta_k \to 0 \) while
\( \tau \cdot \delta_k \to \infty \) as \( \tau \to \infty \).
 Then the conditional median of the post-measurement state satisfies
\[
\operatorname{Med}(X_0 \mid X_\tau \neq -\infty, B = k)
= \log(\epsilon) + \log(2) + \log\left( \tau \cdot \delta_k \right) + R(\tau),
\]
where the remainder term satisfies
$ \displaystyle
\lim_{\tau \to \infty} R(\tau) = 0.
$
\end{theorem}

\noindent
The proof is provided in the appendix. Theorem~\ref{thm:asymptotic_median} shows that the conditional median of \( X_0 \) increases with \( \tau \), with the leading-order behavior captured by \( \log\left( \tau \cdot \delta_k \right) \). Conceptually, the result follows from equation \eqref{ConditionalMedian}, once we account for the additional shift of \( \log(\delta_k) \) in the initial condition \( X_0 \) induced by the measurement step (B) --- though the formal proof is a bit more involved.

\bigskip

In Theorems~\ref{th:BasicModel} and~\ref{th:GeneralizedModel}, we treated $x_{a/b} = \log(\phi_{a/b})$ as fixed while taking the limit $\epsilon \to 0$, i.e., $\log \epsilon \to -\infty$. However, the relevant quantity for the results is not the absolute magnitude of $x_{a/b}$ or $\log \epsilon$ individually, but rather their difference: $x_{a/b} - \log \epsilon$.  
Theorem~\ref{thm:asymptotic_median} shows that, conditional on survival until time $\tau$, the initial difference $X_0 - \log \epsilon = \log(\Phi_0 / \xi_0)$ typically scales as $\log(\tau \cdot \delta_k)$ at $\tau = 0$.  
Thus, if we hold $X_0 = x_{a/b}$ fixed, Theorem~\ref{thm:asymptotic_median} implies that $\epsilon$ typically scales as $1 / (\tau \cdot \delta_k)$. This scaling justifies the limit conditions $\epsilon \to 0$ and $\epsilon \cdot t \to \infty$ used in Theorems~\ref{th:BasicModel} and~\ref{th:GeneralizedModel}, provided that $1\gg \delta_k \gg 1/\tau$.

We stress again that in an actual measurement, the branching ratio $\delta_k$ does not shrink with $\tau$.
But in a serious test of the Born rule (such as the example in the introduction involving 1000 Stern–Gerlach measurements)
the relevant $\delta_k$ values are expected to be small. In that regime, the asymptotic conditions in Theorems~\ref{th:BasicModel} and~\ref{th:GeneralizedModel}
are expected to provide accurate approximations,
justified by Theorem~\ref{thm:asymptotic_median}.

\section{Remarks and speculation on physical implications}

In this section, we offer some remarks and preliminary speculations on physical implications of the Born rule emergence mechanism proposed in this paper.

\subsubsection*{Modification of the Schrödinger equation}

Our proposed small signal truncation mechanism suggests that standard quantum dynamics are only violated when a quantum branch's amplitude falls below a critical threshold.
However, as shown in the previous section,  the typical squared amplitude $\langle \psi | \psi \rangle$ will be well above this threshold (exceeding it by a factor of $\tau$ for the median in \eqref{ConditionalMedian} after applying the exponential transformation) when conditioned on our branch surviving for a macroscopic time $\tau$ after the experiment. Since our ability to observe and discuss the experiment requires this survival condition, it seems highly improbable that we would ever detect any violation  of  the Schrödinger equation in any actual quantum experiment.

\subsubsection*{Conservation laws}

The Schrödinger equation guarantees conservation laws for physical observables through Noether's theorem, which connects symmetries to conserved quantities in the full quantum state vector.
However, due to quantum branching our observed reality follows a single path through the branching tree, and conservation laws will generally be violated when considering only our single branch, which represents just a portion of the quantum state vector.\footnote{For example, when a particle prepared with spin-x up is measured in the y-direction, our branch ends up with definite y-spin where none existed before, apparently creating angular momentum in that direction.}
The Born rule becomes crucial in this context, as it ensures that the probability distribution of measurement outcomes preserves these conservation laws in expectation --- and  thus macroscopically.

Thus, even a slight deviation from the Born rule may lead to observable macroscopic violations of conservation laws.
Our framework is designed to guarantee a very good approximation to the Born rule at macroscopic scales, so conservation laws remain effectively intact.
Nevertheless, our derivations also point to possible deviations from the Born rule (e.g., the coefficient $\beta$ may not be exactly one, and in our discrete time framework we require a small $\epsilon$ approximation). There might also be special physical circumstances (some of which are discussed below) where the assumptions for our derivation of the Born rule could be violated.

\subsubsection*{Violation of the Born rule during the early universe}

It is plausible to conjecture that the universe began in a simple, coherent quantum state, and that the full branching structure of the multiverse only emerged gradually. In particular, the threshold behavior $\xi_t = \epsilon \, \alpha^t$ was derived in Section~\ref{subsec:BasicModelEndogenous} as a steady-state property, and before steady-state was reached, the truncation threshold likely followed a different trajectory.

Thus, during this initial phase of branching, truncation may have been minimal or absent, implying that the Born rule would have been significantly violated. Since the Born rule underpins the regularity of physical laws at macroscopic scales (see our discussion of ``conservation laws'' above), this transition period to the steady state would likely have appeared highly irregular from the perspective of standard physics. It is natural to ask whether early-universe phenomena --- such as cosmological inflation --- may be related to this convergence toward quantum statistical equilibrium.

\subsubsection*{Localized branching in relativistic contexts}

In a relativistic context, 
one may question whether the Schrödinger equation --- and the branching dynamics it governs --- remains valid across all of spacetime (see, e.g.\ \citealt{penrose2006road}). It is plausible that quantum evolution is only well-defined within local regions where spacetime geometry is approximately classical and stable.

Under this view, our framework of quantum branching with small-signal truncation would apply only within such localized regions. The branching processes described in this paper would then represent local structures, rather than global features of the entire universe.

In regions where spacetime is strongly curved, highly dynamical, or fundamentally ill-defined --- such as near spacetime singularities or during the earliest moments of the universe --- the assumptions underlying both the Schrödinger equation and our truncation mechanism may break down. In such regimes, deviations from the Born rule might be expected.

\subsubsection*{Information conservation}

The introduction of small-signal truncation explicitly violates unitarity of the time-evolution operator --- a core principle of quantum mechanics that ensures probability conservation and reversible time evolution. However, this violation need not imply a fundamental loss of information or time reversibility. If the Schrödinger equation is merely an effective continuum approximation to a deeper, discrete physical law, then information could still be exactly preserved at the fundamental discrete level --- much like reversible microscopic dynamics in classical mechanics give rise to apparently irreversible macroscopic behavior.

Furthermore, as we follow our branch of the wave function, information appears to be lost from our subjective perspective — since within our branch, we experience wave function collapse, which seems to violate unitarity. This apparent loss mirrors the black hole information paradox, and indeed some have argued that if information is effectively inaccessible to observers anyway, the paradox may be less troubling within a Many-Worlds framework (see e.g.\ \citealt{wallace2018case}). Within our framework, this perspective on the black hole information paradox is further complemented by the possibility that even the Born rule may not hold universally; in extreme regimes where decoherence fails or branching is disrupted, the conditions for Born rule emergence may break down --- suggesting that the coexistence of apparent information loss and deviations from standard quantum statistics may offer a novel angle on the paradox itself.

\section{Conclusions}

Everett's Many-Worlds Interpretation resolves the most pressing inconsistency of textbook quantum mechanics by eliminating the need for wave function collapse, which conflicts with the dynamics of the Schrödinger equation. However, there still remains a conflict between the Schrödinger equation and the Born rule, because there should not be two separate laws that govern the time evolution of a physical theory, and because a fundamental postulate of the theory should not make statements about emergent phenomena --- which the Born rule does by postulating a probability distribution over branches of the multiverse that are emergent.

In the current paper we have shown that these inconsistencies can be resolved by dropping the Born rule as a fundamental postulate altogether. Instead, the Born rule can be shown to emerge as a consequence of a modified dynamics that adds a mechanism of small-signal truncation to the Schrödinger equation. 
A key difference, however, is that the Born rule that we show to be emergent is not a probabilistic concept, but rather a deterministic statement about branch proliferation in the multiverse --- wave function components with higher amplitudes generate proportionally more future branches than those with lower amplitudes.

This ``world counting'' approach to explaining the Born rule was previously proposed by \citet{hanson2003worlds,hanson2006drift} and related ideas have been explored by \citet{strayhorn2008illustration}. We fully agree with this perspective, which 
 implies a crucial conceptual shift:  Quantum mechanics transforms from an inherently probabilistic theory into a deterministic one where apparent probabilistic behavior emerges from the frequency distribution of future multiverse branches. Any prediction with confidence level $1-\alpha$ in the standard framework (for some $\alpha>0$ near zero) translates to the statement that we should not expect to find ourselves in the small fraction $\alpha$ of future branches where this prediction fails. Occasionally we will observe violations of such predictions, but this should occur rarely if our theory is correct.

An interesting philosophical question in understanding our observed reality is how a single observed path is selected from the vast multiverse branching tree. Our reformulation of the Born rule, however, only requires selecting a ``typical'' path --- in the sense that for the vast majority of non-terminated paths our physical predictions already hold true. In contrast, the standard probabilistic formulation of the Born rule (without small signal truncation) often requires the selected path to be extremely special, as demonstrated by our repeated Stern-Gerlach experiment example in the introduction. It seems fair to conclude that, in our formulation, the remaining selection problem is of philosophical interest but no longer requires further physical explanation.

Our approach to the Born rule is guided by three principles  which are worth emphasizing here:
\begin{enumerate}
    \item \textbf{Unified dynamics:} In an objective framework, all physical theories rest on two core elements: an initial state and a governing law of motion that dictates its evolution. Fundamental physics should, therefore, be expressible in terms of a single rule governing time evolution. This principle of parsimony has driven many major advances in physics, and our approach reinstates it within quantum mechanics.
    
    \item \textbf{Discrete foundations}: If fundamental physics operates on discrete rather than continuous mathematics, the Schrödinger equation must be an approximation, because it fundamentally relies on complex numbers. It is then plausible to expect some form of small-signal truncation of wavefunction amplitudes at the threshold where this approximation ceases to be accurate, analogous to digital computers having finite precision limits. This perspective aligns with various other theoretical approaches suggesting that continuous mathematical structures --- including complex and real numbers --- may not be fundamental to physics. 

     \item \ \textbf{Probability does not exist fundamentally}: Equally, probability is a powerful tool, but one that exists in our models rather than in physical reality itself. As argued by many practitioners of probability theory and statistics, it is implausible that probability represents an intrinsic feature of the world independent of human conception --- it is instead a mathematical construct we impose to describe and manage uncertainty, see e.g.\ \cite{spiegelhalter2024probability}.

\end{enumerate}
These three conceptual foundations crucially shape our treatment of the Born rule in this paper, but we recognize that reasonable minds may differ on these principles. Additionally, alternative derivations of the Born rule might be consistent with these principles while differing significantly in their technical mechanisms.

Even if one accepts the reformulation of the Born rule proposed in this paper, many open questions remain. For example: What is the physcial origin of the small-signal truncation that we take as given and model in a relatively stylized way? Why does the truncation threshold decay exponentially at exactly the rate needed to ensure that the exponent $\beta$ equals one? (Our explanation of this in Section~\ref{subsec:BasicModelEndogenous} should be viewed as a proof of concept rather than a complete account.) Are there other physical implications of the small-signal truncation and the branching process described here, beyond explaining the Born rule? Hopefully, this paper will encourage further discussion of these questions.

\setlength{\bibsep}{5pt}

\bigskip
\bigskip

\section*{Acknowledgments}
\addcontentsline{toc}{section}{Acknowledgments} 

A key inspiration for this paper was David Spiegelhalter's recent Nature essay, ``Why Probability Probably Doesn’t Exist (but It Is Useful to Act Like It Does)'' --- a title we could have easily borrowed here.
I am grateful to my wife, as well as to Oxford University and Nuffield College, for their support --- though all three would likely prefer me to write about Economics. 
I thank Robin Hanson for his valuable insights regarding the relationship between this paper and his work on Mangled Worlds.
I am also sincerely grateful to the professors who taught me Physics years ago, especially Richard Kowarschik and Dietrich Kramer in Jena; Hans Fraas, Thorsten Ohl, and Reinhold Rueckl in Würzburg; and Jan Louis and Henning Samtleben in Hamburg.
Finally, I thank ChatGPT and Claude.ai for helping to improve many formulations in this paper and for providing invaluable research assistance with regards to proofs of various stochastic process results in the appendix.

\appendix
\section*{Appendix}

\addcontentsline{toc}{section}{Appendix A: Comparison with Hanson's ``Mangled Worlds''}
\addcontentsline{toc}{section}{Appendix B: Proofs}

\makeatletter
\pretocmd{\section}{\@nobreaktrue\@afterindentfalse\@afterheading\let\addcontentsline\@gobblethree}{}{}
\pretocmd{\subsection}{\@nobreaktrue\@afterindentfalse\@afterheading\let\addcontentsline\@gobblethree}{}{}
\makeatother

\section{Comparison with Hanson's ``Mangled Worlds''}

\label{app:Hanson}
 
\subsection{Key similarities}

The main conceptual similarities between the ``Mangled Worlds'' framework of \cite{hanson2003worlds,hanson2006drift} and our paper are as follows:

\begin{enumerate}[(i)]
\item
Both approaches seek to replace the probability postulate of standard quantum mechanics by an explanation based on world counts --- i.e.\ the frequency of observed outcomes is explained via the distribution of branching events across many worlds.
In doing so, both approaches derive the Born rule as an emergent phenomenon rather than a fundamental postulate.

\item
Both frameworks invoke a mechanism that truncates branches of small amplitude, which would otherwise dominate the world count. In Hanson's model, the mangling process introduces a lower cutoff in the size distribution of observable branches. In our paper, we postulate a small-signal truncation rule as a fundamental modification to quantum dynamics, with discreteness of the underlying theory as our primary motivation.

\item
By modeling quantum branching dynamics with a suitable diffusion process, both Hanson and our paper show that the Born rule emerges as a mathematical consequence of natural branching and truncation.
\end{enumerate}

\subsection{Translation of Hanson's description into our notation}

\cite{hanson2003worlds,hanson2006drift} models the evolution of world sizes via a diffusion process in the logarithm of the branch amplitude. Specifically, he considers the log-size variable
\[
x = \log m,
\]
where \( m \) denotes the squared norm (Born measure) of the branch amplitude. The distribution of worlds over \( x \) evolves according to a drift–diffusion process with drift velocity \( v \) and diffusion coefficient \( w \).
Specifically, let $\mu_0(x,t)$ denote the world number density over \(x\) at time \(t\). If Hanson had introduced $\mu_0(x,t)$ as a probability density that integrates to one (i.e., $\int \mu_0(x,t) dx =1$), then it would satisfy the standard Fokker-Planck equation:\footnote{
Here, we write $\mu_0$ instead of $\mu$ to avoid confusion with our notation for the drift parameter.	
}
$$
\frac{\partial \mu_0}{\partial t} = v \frac{\partial \mu_0}{\partial x} + \frac{w}{2} \frac{\partial^2 \mu_0}{\partial x^2}
$$
However, \cite{hanson2006drift} introduces $\mu_0(x,t)$ as a density function that also incorporates the exponential growth in the total number of worlds due to branching (and therefore does not integrate to one for any $t$). Therefore, the corresponding PDE that he writes down includes additional terms:
$$
\frac{\partial \mu_0}{\partial t} = v \left( \frac{\partial \mu_0}{\partial x} + \mu_0 \right) + \frac{w}{2} \left( \frac{\partial^2 \mu_0}{\partial x^2} - \mu_0 \right).
$$
This formulation captures both the spatial diffusion of worlds in $x$-space and their overall exponential growth in number over time. In our paper, we do not model this exponential growth over time, as the standard Brownian motion with drift defined in display \eqref{ContinuousTimeProcess} cannot incorporate it directly. However, this omission is irrelevant for the relative count results in Theorems~\ref{th:BasicModel}, \ref{th:GeneralizedModel}, \ref{th:ContinuousModelNoEpsilonLimit}, which focus on the ratio of survival probabilities rather than absolute world counts.

The process described so far is a drift-diffusion process without barrier. By the Mangled Worlds argument, Hanson then imposes an absorbing barrier at some position in \( x \)-space (this is initially a time-dependent barrier, later fixed), corresponding to a lower cutoff on branch amplitudes. Worlds whose log-sizes fall below this threshold are considered mangled and are removed from the observed distribution of world counts.

The resulting drift–diffusion description with an absorbing barrier is equivalent to our continuous-time random walk with negative drift and barrier in display \eqref{ContinuousTimeProcess}. However, there are two possible translations of Hanson's notation into ours, depending on whether one interprets his measure variable $m$ as a rescaled or absolute squared amplitude. The following table summarizes the mapping:

\begin{center}
\renewcommand{\arraystretch}{1.2}
\begin{tabular}{|@{\;}l@{\;}|c|c|c|}
\hline
 & \textbf{Hanson's notation} & \textbf{translation A}  & \textbf{translation B}   \\[-10pt]
 &  &  (rescaled) &   (absolute) \\
\hline
continuous time & $t$ & $\tau$ & $\tau$ \\
\hline
(rescaled) squared amplitude & $m$ & $\alpha^{-\tau} \, \Phi_\tau$ & $\Phi_\tau$ \\
\hline
log of (rescaled) squ. ampl.\ & $x = \log(m)$ & $X_\tau = \log(\alpha^{-\tau} \, \Phi_\tau)$ & $\log \Phi_\tau$ \\
\hline
threshold in log-space & $x_b(t) = (w-v)t - \varepsilon$ & $\log(\epsilon)$ & $\log \epsilon + \tau \log \alpha$ \\
\hline
drift parameter & $v$ & $\mu = \log \alpha + \widetilde{\mu}$ & $\widetilde{\mu}$ \\
\hline
diffusion parameter & $w$ & $\sigma^2$ & $\sigma^2$ \\
\hline
\end{tabular}
\end{center}

\bigskip
To match the mathematical derivation of the Born rule in Hanson's work and in ours, it is important to use ``translation A''. In particular,  our Theorem~\ref{th:ContinuousModelNoEpsilonLimit}, and the subsequent equation~\eqref{GeneralBetaCont} show that the exponent \(\beta\) governing the distribution of surviving worlds is given by \(\beta = \mu / \sigma^2\). To recover the Born rule, we require \(\beta = 1\), which in turn imposes the condition \(\mu = \sigma^2\).
In Hanson's notation, this corresponds to requiring   \(v = w\), which is indeed the critical condition identified in his derivation --- see Subsection~\ref{subsec:HansonV=W} below.

At the beginning of Hanson's papers, \(m\) is introduced as the the squared norm of the branch amplitude, which would correspond to ``translation B''. But as we explain in Section~\ref{sec:BasicModel}, if the diffusion variable is the squared amplitude itself, then a constant threshold in (log) squared amplitude space (i.e., a threshold independent of time) is not compatible with the ongoing process of continuous branching --- a fixed threshold would eventually eliminate all branches, leaving no surviving worlds.
Therefore, as soon as Hanson introduces a fixed threshold (once $v=w$ is achieved) in log-amplitude space, he implicitly transitions from ``translation B'' to ``translation A''.

However, since both our justification of $\mu = \sigma^2$ (in Section~\ref{subsec:BasicModelEndogenous}) and Hanson's justification of $v = w$ (discussed in the next subsection) are ultimately based on criteria that are invariant under a (time-dependent) rescaling of the world amplitude, there is actually no contradiction between the two translations.

\subsection{Hanson's justification of $v = w$ (i.e.\ $\mu = \sigma^2$ in our notation)}
\label{subsec:HansonV=W}

Recall that \(\mu_0(x,t)\) denotes the number density of worlds over the log-size variable \(x = \log m\). Importantly, in Hanson (and in the following) \(\mu_0(x,t)\) is density of worlds \emph{without barrier}. This is because, in his approach, the worlds that cross the barrier continue to exist, they just cannot be observed due to mangling.
The second important density function that Hanson introduces is the (Born) ``measure density'' over \(x\), which is given by
\[
e^x \, \mu_0(x,t).
\]
This is the density of the total Born weight contributed by worlds with log-size between \(x\) and \(x+dx\). 
Hanson then argues that the mangling threshold should be placed at (or near) the \emph{median of this measure distribution}. Based on this criterion, one can then show that $v=w$ (or $v \approx w$).

\subsection{Key differences} 

Having discussed the translation between Hanson's diffusion process and ours, and the argument in Hanson for $v=w$, we are now able to summarize the main differences between the Mangled Worlds approach and our paper:

\begin{enumerate}[(i)]

  \item The Mangled Worlds approach does not modify the conventional Schrödinger dynamics. 
  This implies that the reason for the truncation of worlds counts is not that worlds below the threshold stop to exist (this would violate the Schrödinger equation), but that observers in these small-amplitude worlds become ``mangled'' through interference from larger worlds --- either being destroyed or having their memories altered to match those in larger worlds --- while the world branches themselves continue to exist in the universal wavefunction.

  By contrast, our leading motivation for small-signal truncation stems from the hypothesis that the Schrödinger equation itself may be an approximation of an underlying discrete theory. In our approach, branches with amplitudes below a certain threshold are eliminated entirely from the universal wavefunction, not just rendered unobservable.

  \item Power laws in the tail behavior of
  random variables and stochastic processes are ubiquitous in both theoretical probability and statistical physics --- so e.g.\ our Theorem~\ref{th:BasicModel} is not surprising from that perspcetive.  
  However, to obtain the Born rule we require a very specific power coefficient (equal to one!), which translates into the condition $v=w$ (or $\mu=\sigma^2$ in our notation).

We have briefly summarized Hanson's argument for $v=w$ in Subsection~\ref{subsec:HansonV=W}, while our own argument is presented in Section~\ref{subsec:BasicModelEndogenous}. These approaches differ fundamentally in their reasoning: Hanson's argument for positioning the threshold is based on the ``measure distribution'' that relates to the untruncated branching process, placing the mangling threshold near the median of the total ``measure distribution''. In contrast, our threshold determination, which delivers $\mu=\sigma^2$ is based on the threshold criterion
in equation \eqref{EndogenousThreshold2} and only used the untruncated part of the world distribution.

\item The continuous time diffusion process should be viewed as an approximation of a more complicated underlying branching dynamics. However, while in Hanson this relationship remains implicit, we explicitly start our paper with a discrete branching process, and only introduce the continuous process afterwards as a mathematically useful approximation.

This matters for a holistic understanding of the Born rule emergence, because Theorems~\ref{th:BasicModel} and \ref{th:GeneralizedModel} for the discrete branching processes are more complicated than Theorem~\ref{th:ContinuousModelNoEpsilonLimit} for the continuous one. In particular, the extra condition $\epsilon \rightarrow 0$ (such that $\epsilon \, t \rightarrow \infty$) is crucial for the validity of the discrete branching theorems, while unnecessary for the continuous one. This condition implies that in an actual experiment we can only expect to observe the Born rule when we are sufficiently far from the truncation threshold (small $\epsilon$) --- without this condition, one would expect deviations from the Born rule due to the granularity of the branching process.
Our explicit treatment of this issue in Section~\ref{sec:unified} provides a clearer understanding of why the Born rule emerges in actual experiments even if the underlying branching process is discrete.

\item \cite{hanson2006drift} develops explicit quantitative predictions for deviations from the Born rule in his Section 5. Those could in principle be tested experimentally. Our paper focuses primarily on establishing the mechanism for Born rule emergence but does not derive specific testable deviations, leaving this aspect for future work.

\end{enumerate}

\section{Proofs}

\subsection{Proofs of Theorems \ref{th:BasicModel},  \ref{th:GeneralizedModel},
 \ref{th:ContinuousModelNoEpsilonLimit}}

\subsubsection{Useful results from \cite{karlin1981second}}

For the proof of Theorem~\ref{th:ContinuousModelNoEpsilonLimit}, we rely on classical results from the theory of one-dimensional diffusion processes with absorbing boundaries, as presented in \cite{karlin1981second}, particularly Chapter 15, Section 13, which covers the spectral representation of the transition density. We summarize the key ingredients relevant to our analysis.
We use the notation \( f(\tau) \sim g(\tau) \) to denote that \( f(\tau)/g(\tau) \to 1 \) as \( \tau \to \infty \), and \( f(x) \propto g(x) \) to denote equality up to a constant factor independent of \( x \).

Consider a one-dimensional diffusion process \( \{X(\tau), \tau \geq 0\} \) evolving on an interval \( (c, \infty) \), with an absorbing boundary at \( x = c \), and governed by the stochastic differential equation:
\[
    dX(\tau) = -\mu \, d\tau + \sigma \, dW(\tau),
\]
where \( \mu > 0 \), \( \sigma > 0 \), and \( \{W(\tau)\} \) is a standard Brownian motion. The infinitesimal generator \( L \) of this process is given by
\[
    L = -\mu \frac{d}{dx} + \frac{\sigma^2}{2} \frac{d^2}{dx^2}.
\]
Let \( p(\tau, x, y) \) denote the transition density of the process killed upon hitting the absorbing boundary at \( x = c \). Then, following \cite{karlin1981second}, the transition density\footnote{This is the transition density function with respect to the speed measure. The speed measure \( m(x) \) is a weighting function that arises in the spectral theory of one-dimensional diffusions. It ensures that the generator of the diffusion is self-adjoint in the corresponding \( L^2(m) \) space, and that the eigenfunctions \( \varphi_n(x) \) form an orthonormal basis. For a diffusion with constant drift \( -\mu \) and constant variance \( \sigma^2 \), the speed measure is \( m(x) = \frac{2}{\sigma^2} e^{-2\mu x / \sigma^2} \) (up to normalization).}
admits the spectral representation
\[
    p(\tau, x, y) = m(y) \sum_{n=0}^\infty e^{-\lambda_n \tau} \varphi_n(x) \varphi_n(y),
\]
where \( \{\lambda_n\} \) are the eigenvalues of \( -L \) (with \( 0 < \lambda_0 < \lambda_1 < \cdots \)) and \( \{\varphi_n\} \) are the corresponding eigenfunctions, orthonormal in \( L^2(m) \), and satisfying the boundary condition \( \varphi_n(c) = 0 \).
Define the survival probability:
\[
    q(\tau, x) := \mathbb{P}_x(T > \tau),
\]
where \( T = \inf\{\tau \geq 0 : X(\tau) = c\} \) is the first hitting time of the absorbing boundary, and \( \mathbb{P}_x \) denotes the probability law for the process starting from \( X(0) = x \). From the spectral representation, it follows that for large \( \tau \), the survival probability decays exponentially:
\begin{equation}
    q(\tau, x) \sim \varphi_0(x) e^{-\lambda_0 \tau}, \qquad \text{as } \tau \to \infty,
\label{eq:SurvivalAsymptotics}
\end{equation}
where \( \varphi_0 \) is the eigenfunction corresponding to the smallest eigenvalue \( \lambda_0 \). In particular, the ratio of survival probabilities from two initial positions \( x_a, x_b > c \) satisfies:
\begin{equation}
    \lim_{\tau \to \infty} \frac{q(\tau, x_a)}{q(\tau, x_b)} = \frac{\varphi_0(x_a)}{\varphi_0(x_b)}.
\label{eq:SurvivalRatio}
\end{equation}
In our specific setting, the process has constant drift \( -\mu \), constant variance \( \sigma^2 \), and the absorbing barrier is located at \( x = \log(\epsilon) \). The eigenfunction \( \varphi_0 \) solves the second-order differential equation
\[
    L\varphi_0 = -\lambda_0 \varphi_0,
\]
together with the boundary condition \( \varphi_0(\log(\epsilon)) = 0 \). Solving this ODE explicitly, one obtains a general solution of the form:
\[
    \varphi_0(x) = A e^{\mu x / \sigma^2} + B e^{r_- x}, \quad \text{with } r_- < \mu / \sigma^2,
\]
and the boundary condition determines the constants \( A \) and \( B \). For all \( x > \log(\epsilon) \), the dominant term is the exponential \( e^{\mu x / \sigma^2} \), and the eigenfunction satisfies
\begin{align}
    \varphi_0(x) \propto e^{\mu x / \sigma^2}.
    \label{EigenfunctionDominantTerm}
\end{align}
This establishes equations~\eqref{eq:SurvivalAsymptotics} and \eqref{eq:SurvivalRatio}, which form the basis for our proof of Theorem~\ref{th:ContinuousModelNoEpsilonLimit}.

\bigskip

\subsubsection{Proof of the main text theorems}

The above  results  from \cite{karlin1981second} provide the foundation needed for proving Theorem~\ref{th:ContinuousModelNoEpsilonLimit}.
We will prove Theorem~\ref{th:ContinuousModelNoEpsilonLimit} first,
then us it to prove Theorem~\ref{th:GeneralizedModel}, and then finally use that
to prove Theorem~\ref{th:BasicModel}.
\bigskip

\begin{proof}[Proof of Theorem~\ref{th:ContinuousModelNoEpsilonLimit}]
We consider the continuous-time process $\{X_\tau: \tau \geq 0\}$ defined in \eqref{ContinuousTimeProcess}.
Let $P(\tau, x) := \mathbb{P}(X_\tau \neq -\infty \mid X_0 = x)$ denote the survival probability.
From the spectral representation results summarized in Equations~\eqref{eq:SurvivalAsymptotics} and \eqref{eq:SurvivalRatio}, we know that for large $\tau$, the survival probability admits the asymptotic form
\[
P(\tau, x) \sim \varphi_0(x) e^{-\lambda_0 \tau},
\]
where $\lambda_0$ is the smallest eigenvalue of the generator, and $\varphi_0(x)$ is the corresponding eigenfunction.
Therefore, for any two initial conditions $x_a, x_b > \log(\epsilon)$, we have:
\[
\lim_{\tau \to \infty} \frac{P(\tau, x_a)}{P(\tau, x_b)} = \frac{\varphi_0(x_a)}{\varphi_0(x_b)}.
\]
Substituting the expression for \( \varphi_0(x) \) from \eqref{EigenfunctionDominantTerm}, we obtain
\[
\lim_{\tau \to \infty} \frac{P(\tau, x_a)}{P(\tau, x_b)} = \frac{e^{\frac{\mu}{\sigma^2}x_a}}{e^{\frac{\mu}{\sigma^2}x_b}} = e^{\frac{\mu}{\sigma^2}(x_a - x_b)},
\]
as claimed in the theorem.
\end{proof}

\bigskip

\begin{proof}[Proof of Theorem~\ref{th:GeneralizedModel}]
We prove the result by showing that, under the limit regime $\epsilon \to 0$, $t \to \infty$, and $\epsilon \, t \to \infty$, the discrete-time process defined in \eqref{SochasticProcessRepresentation} converges to the continuous-time process of Theorem~\ref{th:ContinuousModelNoEpsilonLimit}, and that the ratio of survival probabilities converges accordingly.

\medskip
\noindent
\# {\bf Rescaling and functional limit:}
Let $\{X_t^{(\epsilon)}\}_{t \geq 0}$ be the discrete-time process with absorption at $\log(\epsilon)$, and define the rescaled process\footnote{One could introduce a separate time-rescaling parameter \(\delta > 0\) and define the rescaled process as \(Y^{(\delta)}(\tau) := X^{(\epsilon)}_{\lfloor \tau/\delta \rfloor}\), with appropriate asymptotics ensuring \(\epsilon t = (\epsilon/\delta)\tau \to \infty\). However, setting \(\delta = \epsilon\) simplifies notation and is sufficient for the functional limit argument used here.}
\[
Y^{(\epsilon)}(\tau) := X^{(\epsilon)}_{\lfloor \tau/\epsilon \rfloor}, \quad \tau \geq 0.
\]
This accelerates time by a factor of $1/\epsilon$ and enables comparison with continuous-time dynamics.
By the functional central limit theorem for the unconstrained process (e.g., \citealp[Theorem 14.1]{billingsley1999convergence}) and the continuous mapping theorem for stopping times (e.g., \citealp[Theorem 13.6]{billingsley1999convergence}; see also \citealp[Chapter 13]{whitt2002stochastic}), the rescaled process with absorption converges in distribution, as $\epsilon \to 0$, to the diffusion $\{X_\tau\}_{\tau \geq 0}$ solving
\[
dX_\tau = -\mu \, d\tau + \sigma \, dW_\tau, \quad X_0 = x,
\]
with absorption at $\log(\epsilon)$. The assumptions on the i.i.d.\ shocks \(U_t\) ensure that the FCLT applies to the unconstrained process, and the continuity of Brownian sample paths guarantees that hitting times and survival probabilities also converge via the continuous mapping theorem.

\medskip
\noindent
\# {\bf Convergence of survival probabilities:}
Let
\[
P^{(d)}(t, x, \epsilon) := \mathbb{P}(X_t^{(\epsilon)} \neq -\infty \mid X_0^{(\epsilon)} = x),
\quad
P^{(c)}(\tau, x, \epsilon) := \mathbb{P}(X_\tau \neq -\infty \mid X_0 = x).
\]
Note that $X_t^{(\epsilon)} \neq -\infty$ if and only if the process remains above the barrier up to time $t$, and similarly for the continuous-time case. Define the stopping time:
\[
\tau^{(\epsilon)} := \inf\{t \geq 0 : X_t^{(\epsilon)} < \log(\epsilon)\}, \quad
\bar{\tau}^{(\epsilon)} := \epsilon \cdot \tau^{(\epsilon)} = \inf\{\tau \geq 0 : Y^{(\epsilon)}(\tau) < \log(\epsilon)\}.
\]
Then $P^{(d)}(t, x, \epsilon) = \mathbb{P}(\bar{\tau}^{(\epsilon)} > \epsilon \, t \mid X_0^{(\epsilon)} = x)$.
Since $Y^{(\epsilon)} \Rightarrow X$ in the Skorokhod topology and $X$ has continuous sample paths, the map $y \mapsto \inf\{ \tau \leq T : y(\tau) < \log(\epsilon) \}$ is continuous at almost every path of $X$. 
This convergence requires that the initial condition \(x > \log(\epsilon)\) lies sufficiently above the barrier. Specifically, we must assume that
\[
x - \log(\epsilon) \gg \epsilon,
\]
so that the process does not become absorbed immediately in the discrete-time setting, and the FCLT approximation remains valid over the time interval \([0, \tau]\). The same applies to \(x_a\) and \(x_b\) in the survival probability ratio considered below.

Thus, by the continuous mapping theorem (see \citealp[Theorem 13.6]{billingsley1999convergence}; also \citealp[Ch.~13]{whitt2002stochastic}), we obtain:
\[
\lim_{\epsilon \to 0} P^{(d)}(\tau/\epsilon, x, \epsilon) = P^{(c)}(\tau, x, \epsilon), \quad \text{for all fixed } \tau > 0 \text{ and } x > \log(\epsilon).
\]

\medskip
\noindent
\# {\bf Taking the limit:}
Set $\tau := \epsilon \, t$ and consider
\[
\frac{P^{(d)}(t, x_a, \epsilon)}{P^{(d)}(t, x_b, \epsilon)}
= \frac{P^{(d)}(\tau/\epsilon, x_a, \epsilon)}{P^{(d)}(\tau/\epsilon, x_b, \epsilon)}.
\]
By the argument above, this ratio converges to
$
\frac{P^{(c)}(\tau, x_a, \epsilon)}{P^{(c)}(\tau, x_b, \epsilon)}
$ as $\epsilon \to 0$.
Finally, under the condition $\epsilon \, t = \tau \to \infty$, Theorem~\ref{th:ContinuousModelNoEpsilonLimit} gives:
\[
\lim_{\tau \to \infty} \frac{P^{(c)}(\tau, x_a, \epsilon)}{P^{(c)}(\tau, x_b, \epsilon)} = \exp\left[ \frac{\mu}{\sigma^2}(x_a - x_b) \right],
\]
and we thus obtain
\[
\lim_{0 \ll \epsilon^{-1} \ll t}
\frac{P^{(d)}(t, x_a, \epsilon)}{P^{(d)}(t, x_b, \epsilon)}
= \exp\left[ \frac{\mu}{\sigma^2}(x_a - x_b) \right],
\]
as stated in the theorem.
\end{proof}

\bigskip

\begin{proof}[Proof of Theorem \ref{th:BasicModel}]
   Theorem~\ref{th:BasicModel} is a special case of Theorem~\ref{th:GeneralizedModel} and its implication in equation \eqref{GeneralBeta}.  
   At the beginning of Section~\ref{sec:General}, we have already rewritten the model from Theorem~\ref{th:BasicModel} in the stochastic process notation required for Theorem~\ref{th:GeneralizedModel}, but we still need to verify the assumptions of Theorem~\ref{th:GeneralizedModel}.
From equation~\eqref{BasicModelParametersShocks}, we have 
   \begin{align*}
      \mu &= \log\left(\alpha / \overline{\delta}\right), &
      \sigma^2 &= \frac{1}{K}\sum_{k=1}^K\left[\log\left(\delta_k / \overline{\delta}\right)\right]^2, &
      U_t &= \frac{1}{\sigma}\log\left(\delta_{B_t} / \overline{\delta}\right),
   \end{align*}
   where $\overline{\delta} = \left(\prod_{k=1}^K \delta_k\right)^{1/K}$ and $B_t$ is uniformly distributed on $\{1,\ldots,K\}$.
   We now verify that all assumptions of Theorem~\ref{th:GeneralizedModel} are satisfied:
   \begin{itemize}
      \item $\mu > 0$ follows from assumption (ii) which states $\overline{\delta} < \alpha$, implying $\log(\alpha/\overline{\delta}) > 0$.
      
      \item $\sigma > 0$ follows, because  assumption (i)
      implies that not all $\delta_k$ can be identical.
      
      \item $\mathbb{E}[U_t] = 0$ and $\mathbb{E}[U_t^2] = 1$ by construction.
      
      \item Since $U_t$ takes only finitely many values, $\mathbb{E}[|U_t|^{2+\gamma}] < \infty$ for any $\gamma > 0$.
      
      \item $\mathbb{P}(U_t > \mu/\sigma) > 0$ follows from assumption (ii) which states $\alpha < \max_k \delta_k$, implying there exists at least one $k$ such that $\delta_k > \alpha$, which ensures $\mathbb{P}(U_t > \mu/\sigma) > 0$.
      
      \item Assumption (i) guarantees that the distribution of $U_t$ is non-lattice. %
   \end{itemize}
   Having verified all conditions, we can use Theorem~\ref{th:GeneralizedModel} and its implication in equation \eqref{GeneralBeta} to obtain the desired result.
\end{proof}

\bigskip

\subsection{Proof of Theorems \ref{th:StationaryDistribution} and \ref{th:EndogenousThresholdAsymptotics}}

\begin{proof}[Proof of Theorem~\ref{th:StationaryDistribution}]
\# {\bf Shift the process:} Let \( Y_\tau := X_\tau - \log(\epsilon) \), so that \( Y_\tau \) evolves on \([0,\infty)\) and satisfies
$$
dY_\tau = -\mu \, d\tau + \sigma \, dW_\tau, \quad \text{with absorption at } 0.
$$

\medskip
\noindent
\# {\bf Doob's \( h \)-transform:} To condition on survival, we apply Doob’s \( h \)-transform (see e.g.\ \citealt{karatzas1991brownian}, Chapter 5.4) using a positive harmonic function \( h \) solving \( L \,h = 0 \), where \( L = -\mu \frac{d}{dy} + \frac{\sigma^2}{2} \frac{d^2}{dy^2} \) is the infinitesimal generator of \( Y_\tau \). The exponential function \( h(y) = e^{\frac{2\mu}{\sigma^2} y} \) satisfies this equation.

\medskip
\noindent
\# {\bf Conditioned process and its stationary distribution:} 
Under the \( h \)-transform, the conditioned process has generator
\[
L^h f(y) = \mu \frac{d}{dy} f(y) + \frac{\sigma^2}{2} \frac{d^2}{dy^2} f(y),
\]
corresponding to a Brownian motion with drift \( +\mu \) and diffusion coefficient \( \sigma \). As is well known (see, e.g., \citealt[Chapter 15]{karlin1981second}), this process admits a unique stationary distribution on \( [0, \infty) \), which is exponential with rate \( \mu / \sigma^2 \).

\medskip
\noindent
\# {\bf Shift back:} Returning to \( X_\tau = Y_\tau + \log(\epsilon) \), the limiting conditional distribution of \( X_\tau \) has density
\[
f_X(x) = \frac{\mu}{\sigma^2} \exp\left\{ -\frac{\mu}{\sigma^2}(x - \log(\epsilon)) \right\}, \quad x \geq \log(\epsilon).
\]
This is what we wanted to show.
We note that the harmonic function above has exponential rate \( 2\mu/\sigma^2 \), but the stationary density of the conditioned process has rate \( \mu/\sigma^2 \). This reflects the difference between the reweighting function in the \( h \)-transform and the long-run behavior of the resulting (drifted) process.
\end{proof}

\bigskip

\begin{proof}[Proof of Theorem~\ref{th:EndogenousThresholdAsymptotics}]
We study the asymptotic growth of the self-consistent threshold
\[
\xi_\tau := \varepsilon \cdot \mathbb{E}[\Phi_\tau \mid \Phi_\tau > \xi_\tau],
\]
where \(\Phi_\tau\) evolves according to
\[
d \log \Phi_\tau = -\widetilde{\mu} \, d\tau + \sigma \, dW_\tau, \quad \Phi_0 = \phi_0 > 0,
\]
and is absorbed when $\Phi_\tau \leq \xi_\tau$.
Let \( Z_\tau := \log \Phi_\tau \) and \( b_\tau := \log \xi_\tau \). Then \( Z_\tau \) follows a Brownian motion with drift $-\widetilde{\mu}$ and variance $\sigma^2$, absorbed at the moving barrier \( b_\tau \), and satisfies the self-consistency equation:
\[
b_\tau = \log\left( \varepsilon \cdot \mathbb{E}[e^{Z_\tau} \mid Z_\tau > b_\tau] \right).
\]

\medskip
\noindent
\# {\bf Exponential growth Ansatz:}
Assume
\[
\xi_\tau = c_0 \cdot \alpha^\tau \cdot (1 + o(1)), \quad \text{so that} \quad b_\tau = \log c_0 + \tau \log \alpha + o(\tau),
\]
for constants \( c_0 > 0 \), \( \alpha > 0 \) to be determined.

\medskip
\noindent
\# {\bf Quasi-stationary approximation:}
For large $\tau$, the conditional law of \( Z_\tau - b_\tau \) given survival converges to a quasi-stationary distribution (QSD). The QSD for a drifted Brownian motion absorbed at a fixed boundary \( b \) has exponential density:
\[
\nu_b(z) \propto e^{-\eta(z-b)} \quad \text{for } z > b,
\]
with decay parameter $\eta$ related to the principal eigenvalue $\lambda_0$ of the infinitesimal generator :
\[
L = -\widetilde{\mu} \frac{d}{dz} + \frac{\sigma^2}{2} \frac{d^2}{dz^2}, \quad \eta = \frac{-\widetilde{\mu} + \sqrt{\widetilde{\mu}^2 + 2 \sigma^2 \lambda_0}}{\sigma^2}.
\]
Using the QSD, we compute:
\[
\mathbb{E}[e^{Z_\tau} \mid Z_\tau > b_\tau] = e^{b_\tau} \cdot \int_0^\infty e^y \cdot \eta e^{-\eta y} dy = e^{b_\tau} \cdot \frac{\eta}{\eta - 1}, \quad \text{provided } \eta > 1.
\]

\medskip
\noindent
\# {\bf Self-consistency and solution for $\eta$:}
Substituting into the self-consistency condition gives:
\[
b_\tau = \log\left( \varepsilon \cdot e^{b_\tau} \cdot \frac{\eta}{\eta - 1} \right) = b_\tau + \log\left( \varepsilon \cdot \frac{\eta}{\eta - 1} \right),
\]
so we must have:
\[
\varepsilon \cdot \frac{\eta}{\eta - 1} = 1 \quad \Rightarrow \quad \eta = \frac{1}{1 - \varepsilon},
\]
which indeed satisfies $\eta > 1$.

\medskip
\noindent
\# {\bf Solving for $\alpha$ and $c_0$:}
From the eigenvalue relation for $\eta$, we have:
\[
\frac{1}{1 - \varepsilon} = \frac{-\widetilde{\mu} + \sqrt{\widetilde{\mu}^2 + 2 \sigma^2 \lambda_0}}{\sigma^2} \quad \Rightarrow \quad \lambda_0 = \frac{\sigma^2}{2(1 - \varepsilon)^2} + \frac{\widetilde{\mu}}{1 - \varepsilon}.
\]
This principal eigenvalue governs the survival decay rate and thus the asymptotic growth of the conditional expectation. The growth rate of \( b_\tau \) is then:
\[
\log \alpha = \frac{d b_\tau}{d \tau} = \frac{\sigma^2}{1 - \varepsilon} - \widetilde{\mu}.
\]
To determine \( c_0 \), recall that initially \( \Phi_0 = \phi_0 \) and thus \( \xi_0 = \varepsilon \cdot \phi_0 \). Matching with the Ansatz \( \xi_0 = c_0 \cdot \alpha^0 \) gives:
\[
c_0 = \frac{\varepsilon \cdot \phi_0}{1 - \varepsilon}.
\]

\medskip
\noindent
\# {\bf Uniqueness and conclusion:}
The map \( \xi \mapsto \varepsilon \cdot \mathbb{E}[\Phi_\tau \mid \Phi_\tau > \xi] \) is continuous and strictly increasing, ensuring uniqueness of the fixed point \( \xi_\tau \) at each $\tau$.
Hence, we conclude:
\[
\xi_\tau = c_0 \cdot \alpha^\tau \cdot (1 + o(1)), \quad \text{with} \quad \log \alpha = \frac{\sigma^2}{1 - \varepsilon} - \widetilde{\mu},
\]
so that:
\[
\lim_{\tau \to \infty} \frac{\log(\xi_\tau / c_0)}{\tau} = \log \alpha,
\]
as claimed.
\end{proof}

\bigskip

\subsection{Proof of the conditional median results in Section~\ref{sec:unified}}

\begin{proof}[Proof of Equation~\eqref{ConditionalMedian}]
We consider the continuous-time process  in \eqref{ContinuousTimeProcess},
assuming that it has reached its steady state distribution in Theorem~\ref{th:StationaryDistribution} by time $\tau = 0$.
It is well known (see, e.g., \cite{karlin1981second}, Chapter 15) that for this process, the survival probability up to time $ \tau $ for large $ \tau $ is asymptotically proportional to the distance from the boundary, with a coefficient that depends on both $ \sigma $ and $ \mu $.
Let $ f_X(x) $ denote the stationary distribution of the process, as established in Theorem~\ref{th:StationaryDistribution}:
\[
f_X(x) = \frac{2\mu}{\sigma^2} \exp\left(-\frac{2\mu(x - \log(\epsilon))}{\sigma^2}\right), \quad x \geq \log(\epsilon).
\]
The survival probability from initial value $ x $ for large $ \tau $ is:
\[
p_\tau(x) \sim \frac{2(x - \log(\epsilon))}{\sigma\sqrt{2\pi\tau}} \exp\left(-\frac{\mu^2\tau}{2\sigma^2}\right), \quad \text{as } \tau \to \infty.
\]
Here, \( \sim \) means that the ratio of the left-hand side to the right-hand side converges to 1 as \( \tau \to \infty \); that is, the expression gives the leading-order asymptotic behavior of the survival probability for large \( \tau \).
Using Bayes' rule, the survival-conditioned density of $ X_0 $ is:
\begin{align*}
f_{X_0 \mid X_\tau \neq -\infty}(x) &= \frac{f_X(x) \cdot p_\tau(x)}{\int_{\log(\epsilon)}^{\infty} f_X(y) \cdot p_\tau(y) \, dy} \\
&\sim \frac{2\mu}{\sigma^2} \exp\left(-\frac{2\mu(x - \log(\epsilon))}{\sigma^2}\right) \cdot \frac{2(x - \log(\epsilon))}{\sigma\sqrt{2\pi\tau}} \cdot \frac{1}{Z_\tau}
\end{align*}
where $Z_\tau$ is the normalization constant. After simplification and collecting terms:
\[
f_{X_0 \mid X_\tau \neq -\infty}(x) \sim C_\tau \cdot (x - \log(\epsilon)) \exp\left(-\frac{2\mu(x - \log(\epsilon))}{\sigma^2}\right)
\]
This corresponds to a shifted Gamma distribution with shape parameter 2, rate parameter $2\mu/\sigma^2$, and shift $\log(\epsilon)$. Specifically, it is a $\text{Gamma}(2, 2\mu/\sigma^2)$ distribution shifted by $\log(\epsilon)$.
The cumulative distribution function of this shifted Gamma distribution is:
\[
F(x) = 1 - \left(1 + \frac{2\mu(x - \log(\epsilon))}{\sigma^2} \right) \exp\left(-\frac{2\mu(x - \log(\epsilon))}{\sigma^2} \right).
\]
The median $m$ satisfies $F(m) = 1/2$, leading to the implicit equation:
\[
\left(1 + \frac{2\mu(m - \log(\epsilon))}{\sigma^2} \right) \exp\left(-\frac{2\mu(m - \log(\epsilon))}{\sigma^2} \right) = \frac{1}{2}.
\]
Setting $z = \frac{2\mu(m - \log(\epsilon))}{\sigma^2}$, this becomes:
$
(1 + z) e^{-z} = \frac{1}{2}.
$
For large $\tau$, the solution to this equation can be obtained using the Lambert W function, but we can also derive an asymptotic approximation. For large \( \tau \), the conditional distribution spreads, and the median \( m \) grows logarithmically. Setting \( z = \frac{2\mu(m - \log(\epsilon))}{\sigma^2} \), the defining equation becomes
\[
(1 + z) e^{-z} = \frac{1}{2}.
\]
Solving asymptotically, we find \( z \sim \log(2\tau\mu/\sigma^2) \), leading to
\[
m \sim \log(\epsilon) + \frac{\sigma^2}{2\mu} \log\left(2\tau\mu/\sigma^2\right).
\]
We thus obtain
$$
\operatorname{Med}(X_0 \mid X_\tau \neq -\infty) = \log(\epsilon) + \log(2) + \frac{\sigma^2}{\mu} \log(\tau) + o(1),
$$
which is what we wanted to show.
\end{proof}

\bigskip

\begin{proof}[Proof of Theorem~\ref{thm:asymptotic_median}]
\# {\bf Initial distribution:}
From condition (A), we know that $X^{\text{past}}_0$ has the stationary distribution with density
\[
f_{X^{\text{past}}_0}(x) = \frac{1}{\sigma^2} \exp\left( -\frac{x - \log(\epsilon)}{\sigma^2} \right), \quad x \geq \log(\epsilon),
\]
where we used the fact that $\mu/\sigma^2 = 1$. 
Under the conditioning \( B = k \), the initial value is given by
\[
X_0 = X^{\text{past}}_0 + \log(\delta_k),
\]
so the density of \( X_0 \) is
\[
f_{X_0}(x) = \frac{1}{\sigma^2} \cdot \mathbf{1}_{x \geq \log(\delta_k \epsilon)} \cdot \exp\left( -\frac{x - \log(\delta_k \epsilon)}{\sigma^2} \right).
\]

\medskip
\noindent
\# {\bf Survival probability:}
Let \( p_\tau(x) \) denote the survival probability from initial condition \( x \). For the Brownian motion with drift described in \eqref{ContinuousTimeProcess}, this probability for large \( \tau \) is:
\[
p_\tau(x) \sim \frac{2(x - \log \epsilon)}{\sigma\sqrt{2\pi\tau}} \exp\left(-\frac{\mu^2\tau}{2\sigma^2}\right), \quad \text{for } x > \log \epsilon.
\]
Using the condition $\mu/\sigma^2 = 1$, and therefore $\mu = \sigma^2$, this simplifies to:
\[
p_\tau(x) \sim \frac{2(x - \log \epsilon)}{\sigma\sqrt{2\pi\tau}} \exp\left(-\frac{\sigma^4\tau}{2\sigma^2}\right) = \frac{2(x - \log \epsilon)}{\sigma\sqrt{2\pi\tau}} \exp\left(-\frac{\sigma^2\tau}{2}\right).
\]

\medskip
\noindent
\# {\bf Conditional density:}
Using Bayes' rule, the survival-conditioned density of $X_0$ given both $X_\tau \neq -\infty$ and $B = k$ is:
\begin{align*}
f_{X_0 \mid X_\tau \neq -\infty, B=k}(x) &= \frac{f_{X_0}(x) \cdot p_\tau(x)}{\int_{\log(\delta_k\epsilon)}^{\infty} f_{X_0}(y) \cdot p_\tau(y) \, dy} \\
&\propto f_{X_0}(x) \cdot p_\tau(x) \\
&\propto \exp\left( -\frac{x - \log(\delta_k \epsilon)}{\sigma^2} \right) \cdot (x - \log \epsilon) \\
&= (x - \log \epsilon) \cdot \exp\left( -\frac{x - \log(\delta_k \epsilon)}{\sigma^2} \right),
\end{align*}
for $x \geq \log(\delta_k\epsilon)$.
 Here, \( \propto \) indicates proportionality up to a constant factor independent of \( x \); that is, we are describing the shape of the density function before normalization.
Next, we define $y = x - \log(\epsilon)$ as the excess above the absorption threshold. Then our density becomes:
\begin{align*}
f_Y(y) &\propto y \cdot \exp\left( -\frac{y - \log(\delta_k)}{\sigma^2} \right) \\
&= y \cdot \exp\left(-\frac{y}{\sigma^2}\right) \cdot \exp\left(\frac{\log(\delta_k)}{\sigma^2}\right) \\
&\propto y \cdot \exp\left(-\frac{y}{\sigma^2}\right) \cdot \delta_k^{1/\sigma^2},
\end{align*}
for $y \geq \log(\delta_k)$.

\medskip
\noindent
\# {\bf Asymptotic median:}
We now determine the asymptotic behavior of the median of the variable \( Y = X_0 - \log(\epsilon) \), whose conditional density takes the form
\[
f_Y(y) \propto y \cdot \exp\left(-\frac{y}{\sigma^2}\right), \quad y \geq \log(\delta_k).
\]
This is a truncated Gamma(2, \( 1/\sigma^2 \)) distribution, shifted by \( \log(\delta_k) \). The truncation is negligible in the asymptotic regime \( \tau \to \infty \) under the assumption \( \tau \cdot \delta_k \to \infty \), since this implies \( \log(\delta_k) = o(\log(\tau)) \). Therefore, the median of \( Y \) converges to that of the full Gamma(2, \( 1/\sigma^2 \)) distribution.
The median \( m \) of a Gamma(2, \( 1/\sigma^2 \)) distribution satisfies the implicit equation
\[
\left(1 + \frac{m}{\sigma^2}\right) e^{-m/\sigma^2} = \frac{1}{2},
\]
which yields the asymptotic solution
\[
m = \sigma^2 \log(2) + \sigma^2 \log\left(\tau \cdot \delta_k\right) + o(1),
\]
reflecting the fact that the mass of the distribution is concentrated at \( y = O(\log(\tau)) \).
Returning to the original variable \( X_0 = Y + \log(\epsilon) \), we conclude that
\[
\operatorname{Med}(X_0 \mid X_\tau \neq -\infty, B = k) = \log(\epsilon) + \log(2) + \log(\tau \cdot \delta_k) + R(\tau),
\]
where \( R(\tau) \to 0 \) as \( \tau \to \infty \).

\medskip
\noindent
\# {\bf Simplification:}
Using the assumption \( \mu/\sigma^2 = 1 \), we have \( \sigma^2 = \mu \), and thus the expression above is equivalent to
\[
\operatorname{Med}(X_0 \mid X_\tau \neq -\infty, B = k) = \log(\epsilon) + \log(2) + \log(\tau \cdot \delta_k) + R(\tau),
\]
as stated in the theorem.
\end{proof}

\end{document}